\documentclass[11pt,letterpaper,twoside,reqno,nosumlimits]{amsart}

\usepackage{etoolbox}

\patchcmd{\section}{\scshape}{\bfseries}{}{}
\makeatletter
\renewcommand{\@secnumfont}{\bfseries}
\makeatother

\usepackage[dvipsnames]{xcolor}
\patchcmd{\section}{\normalfont}{\normalfont\color{MidnightBlue}}{}{}
\patchcmd{\subsection}{\normalfont}{\normalfont\color{MidnightBlue}}{}{}

\makeatletter
\def\subsubsection{\@startsection{subsubsection}{3}%
\z@{.5\linespacing\@plus.7\linespacing}{-.5em}%
{\normalfont\bfseries}}
\makeatother

\usepackage{fancyhdr}
\usepackage{amsmath,amsfonts,amsbsy,amsgen,amscd,mathrsfs,amssymb,amsthm,mathtools,tensor}

\usepackage{amsthm}
\usepackage{amsmath}
\usepackage{amssymb}
\usepackage{dsfont}

\usepackage{amsfonts}
\usepackage{graphicx}
\usepackage{color}
\usepackage[bf,SL,BF]{subfigure}
\usepackage{url}
\usepackage{epstopdf}
\usepackage{pdfsync}
\usepackage[colorlinks]{hyperref}
\usepackage{comment}
\usepackage{mathabx}
\usepackage{upgreek}

\hypersetup{
    linkcolor=blue,
}

\newlength{\fixboxwidth}
\usepackage{epsfig,amsbsy,graphicx,multirow}

\usepackage{ algorithm, algorithmic}

\renewcommand{\algorithmiccomment}[1]{\bgroup\hfill//~#1\egroup}

\usepackage[all]{xy}
\usepackage{accents}

 \numberwithin{equation}{section}

\def\M{\mathcal{M}}

\def\<{\big\langle}
\def\>{\big\rangle}

\definecolor{red}{rgb}{0.9, 0, 0}

\newcommand\blfootnote[1]{%
  \begingroup
  \renewcommand\thefootnote{}\footnote{#1}%
  \addtocounter{footnote}{-1}%
  \endgroup
}

\usepackage[colorinlistoftodos]{todonotes}
\newtheorem{theorem}{Theorem}
\newtheorem{proposition}{Proposition}
\newtheorem{corollary}{Corollary}
\newtheorem{lemma}{Lemma}
\theoremstyle{definition}

\newtheorem{definition}{Definition}

\newtheorem{example}{Example}

\newtheorem{axiomN}{Axiom}[section]
\theoremstyle{remark}

\newtheorem{remark}{Remark}

\usepackage[
    backend=biber,natbib,
    style=authoryear,
    citetracker=true,
    maxcitenames=1
]{biblatex}

\AtEveryCitekey{\ifciteseen{}{\defcounter{maxnames}{99}}}

\begin{document}

\title[Aggregation of models, choices, beliefs, and preferences]{Aggregation of models,\\ choices, beliefs, and preferences}
\author{Hamed Hamze Bajgiran, Houman Owhadi} 

\date{\today}
\blfootnote{Affiliation: Division of Computing and Mathematical Sciences (CMS), California Institute of Technology. Email: \href{mailto:hhamzeyi@caltech.edu}{hhamzeyi@caltech.edu} and \href{mailto:owshadi@caltech.edu}{owhadi@caltech.edu} }

\maketitle

%\chead[Hamed, Houman]{ Aggregation}

%\centerline{\href{https://drive.google.com/file/d/1wYi_rzPG_1mWQrh7BpIdHCjXkLmOZSiz/view?usp=sharing}{Click here for the most recent version}}

%%%%%%%%%%%%%%%%%%%%    Introduction
%%%%%%%%%%%%%%%%%%%%
%%%%%%%%%%%%%%%%%%%%

%\centerline{\href{https://drive.google.com/file/d/1wYi_rzPG_1mWQrh7BpIdHCjXkLmOZSiz/view?usp=sharing}{Click here for the most recent version}}

%\begin{abstract}

%In this paper, we present a recursive model of aggregation. We show that recursive aggregation lies behind many different results in economic theory. We confirm that, whenever the outcome of a model over a set of features can be recursively computed by (1) breaking the set of features into two disjoint subsets, and (2) the aggregated outcome being a weighted average of each of the two smaller subsets, then the model has a simple form:

%There exist a weight function and a ranking order over the set of features, and the outcome conditional on aggregation of a subset of features is the weighted average of outcomes associated with each highest-ordered feature separately.

%The result unifies aggregation procedures across many different economic environments, showing that all of them rely on the same basic result.
%Following the main representation, we show applications and extensions of our representation in various separated topics of economics such as Belief Formation, Choice Theory, and Social Welfare Economics.

%\end{abstract}
%\newpage
%\centerline{\href{https://drive.google.com/file/d/1wYi_rzPG_1mWQrh7BpIdHCjXkLmOZSiz/view?usp=sharing}{Abstract2}}
\begin{abstract}
A  natural notion of rationality/consistency for aggregating models is that, for all (possibly aggregated) models  $A$ and $B$, if the output of model $A$ is $f(A)$ and if the output model $B$ is $f(B)$, then the output of the model obtained by aggregating $A$ and $B$ must be a weighted average of $f(A)$ and $f(B)$.
Similarly, a natural notion of rationality for aggregating preferences of ensembles of experts is that, for all (possibly aggregated) experts  $A$ and $B$, and all possible choices $x$ and $y$, if both $A$  and $B$ prefer $x$ over $y$, then the expert obtained by aggregating $A$ and $B$ must also prefer $x$ over $y$.
Rational aggregation is an important element of uncertainty quantification, and it lies behind many seemingly different results in economic theory: spanning social choice, belief formation, and individual decision making. Three examples of rational aggregation rules are as follows. (1) Give each individual model (expert) a weight (a score) and use weighted averaging to aggregate individual or finite ensembles of models (experts). (2) Order/rank individual model (expert) and let the aggregation of a finite ensemble of individual models (experts) be the highest-ranked individual model (expert) in that ensemble.
(3) Give each individual model (expert) a weight, introduce a weak order/ranking over the set models/experts (two models may share the same rank), aggregate  $A$ and $B$ as the weighted average of the highest-ranked models (experts) in $A$ or $B$. Note that (1) and (2) are particular cases of (3) (in (1) all models/experts share the same rank, and in (2) the ranking is strict).
In this paper, we show that all rational aggregation rules are of the form (3).
This result unifies aggregation procedures across many different economic environments, showing that they all rely on the same basic result.
Following the main representation, we show applications and extensions of our representation in various separated economics topics such as belief formation, choice theory, aggregation of optimal models, and social welfare economics.
\end{abstract}

\section{Introduction}
This paper presents a general framework  
characterizing 
rational/consistent aggregation (of models, choices, beliefs, and preferences, which we simply refer to as features) with applications to economic theory.
In this framework, individual features have outcomes, and aggregation rules identify the outcome of groups of features. 
 We focus on a recursive form of aggregation, which is the case in the Cased-Based decision theory developed by \cite{gilboa_inductive, billot1}, where the aggregate outcome for larger collections of features results from aggregating the outcomes of smaller subsets. Specifically, the aggregate outcome of the union of two disjoint collections of features is a weighted average of the outcome of each collection of features separately. We show that this form of recursive aggregation is a common structure that lies behind many seemingly unrelated results in economic theory. 

%To describe this, we consider a set of features and a set of conditional outcomes. We associate each conditional outcome with a subset of features, representing the outcome of the model conditional on the aggregation of that subset of features.

Our central axiom, the \emph{weighted averaging axiom/property} (we will use the terms \emph{axiom} and \emph{property} interchangeably), is a simple formalization of the recursivity. It imposes a structure on how the outcome of the union of two disjoint subsets of features relates to the outcome of each of the subsets separately. The axiom states that the outcome of a set of features can be recursively computed by first partitioning the set of features into two disjoint subsets. Then, the aggregated outcome is a weighted average of the outcome of each of the two smaller subsets.

Our contribution is three-fold: (1) We find all aggregation procedures that satisfy the weighted averaging axiom, which generalized the result of \cite{billot1}. Moreover, by enhancing the procedure with continuity axiom, we connect the axiom to the path independent axiom, which is studied in the choice literature. (2) With a simple geometrical duality argument, we connect the weighted averaging to the combination axiom of \cite{gilboa_inductive} and Extended Pareto of \cite{shapley2}. (3) We present applications and extensions to different domains of economics, notably in the context of Belief Formation, Choice Theory, and Welfare Economics.

%This paper makes three-fold contributions. The first contribution is to find all aggregation procedures that satisfy the weighted averaging axiom, which generalized the result of \cite{billot1}. Moreover, with enhancing the procedure with continuity axiom we connect our axiom to the path independent axiom studied in the choice literature. The second contribution is to use the simply geometrical duality and connect the the weighted averaging to the combination axiom of \cite{gilboa_inductive} and Extended Pareto of \cite{shapley1}. And the third contribution is to to present applications and extensions to different domains of economics, notably in the context of Belief Formation, Choice Theory, and Welfare Economics.   

%This result unifies aggregation procedures across different economic environments and shows that all of them can be recovered as a consequence of the weighted averaging axiom. The second contribution is 

%We model the dependency of the outcome on the set of aggregated features through what we call an \emph{aggregation rule}.

Formally, we define an aggregation rule as a function on the set of subsets of features that maps each subset of features to an outcome. 
Our main result finds all aggregation rules that satisfy recursivity in the form of our weighted average axiom. We show that as long as for any two disjoint subsets of features, the outcome of their union is a weighted average (with non-negative weights) of the outcome of each subset, then the aggregation rule has a simple form (with a technical richness condition):

\emph{There exist a strictly positive weight function and a weak order (a transitive and complete order) over the set of features, with the outcome of any subset of features being the weighted average of the outcomes of each of the highest-ordered features of the subset separately.}

The importance of the result is that the weight of each feature is independent of the group of features being aggregated. The role of the weak order in the main representation is to partition the set of features into different equivalence classes and rank them from the highest class to the lowest class. If all features of a subset of features are in the same class, then the outcome is the weighted average of the outcomes of each member of the set. However, if some features have a higher ranking than others, then the aggregation rule will ignore lower-ordered features. 

Following the main result, we discuss two special cases of our main result.
%%%%%%%%%%% special case : Strict weighted averaging
In the first case, we introduce the \emph{strict weighted averaging} axiom to represent the case where the outcome of the union of two disjoint subsets of features is contained in the ``relative'' interior of the outcomes of each subset separately. We then show that the strict weighted averaging axiom is the necessary and sufficient condition for the weak order, in the main representation, to have only one equivalence class. Hence, the outcome of a subset of features is just the weighted average (with strictly positive weights) of the outcomes of each feature separately. 

%The strict weighted averaging axiom/property captures the idea that the model cannot ignore the role of one of the subsets in finding the outcome of their union.

%\redt{FE: it is not clear to me that it is worth discussing both results in the introduction. What are the different ideas captured by the two versions of your axiom? Do they seek to capture the same idea, in different ways? If so it's probably enough to talk about only one.}

%%%%%%%%%%% Special case: continuity

In the second case,  we model the space of features as a subset of a vector space. By considering the distance between vectors, we capture the notion of similarity or closeness of features. In this context, we can consider the following  notion of continuity of outcomes with respect to features:
replacing a feature in a subset of features with another closely similar feature, the outcome of this new subset stays close to the outcome of the previous one. Under this property, which we define as the \emph{continuity} property, we show that all similar enough features attain the same ranking with respect to the weak order. Moreover, the weight function is a continuous function over the set of features. In other words, the weight between two close (or similar) features should be close. In a special case, where the space of features is a convex set, we show that all features attain the same ranking. In this case, there is no difference between the weighted averaging and the strict weighted averaging property.

%\todo[inline=yes]{I can remove the next paragraph}
%The result can be interpreted as an impossibility result. It specifies that for an aggregation rule (under a minor richness condition) that satisfies the weighted averaging axiom, it is impossible to satisfy the continuity and to have a value over a subset of features that ignores the value of one of its feature. We explore the implication of this result in the context of choice theory.
%\redt{FE: it is also not clear that this belongs here. Maybe connected to the Kalai result.}

%\redt{The application to conditional expectation is the simplest one, but it is also the one that is less appealing.}

Depending on the application,  features and the aggregation rules may have different interpretations. A feature may represent a signal or an event containing some information about the true state of nature. In this case, the role of an aggregation rule is to form a belief about the true state of nature. In the context of choice theory, features may represent choice objects, where an aggregation rule behaves as a decision-maker that selects a lottery or a random choice out of a group of choice objects. Another interpretation is in the context of welfare economics, where each feature represents a preference of an individual over some alternatives. In this case, an aggregation rule represents a social welfare function that associates with each preference profile, a single preference ordering over the set of alternatives.

To describe a natural interpretation of our result, consider the problem of modeling an agent who seeks to make a prediction about the true state of nature, conditional on observing a set of events. In this context, a feature represents an event, and the outcome of the model conditional on observing a set of events is a belief about the true state of nature. Our main result provides a necessary and sufficient condition for the belief formation process to behave as a \emph{Bayesian Updater}. Under the averaging property of the belief formation process, there exists a conditional probability system associated with the set of events, and the belief formation process conditional on observing a set of events behaves like a conditional probability. The weak order of the main result is capturing the idea that conditional on observing even a zero probability event, the belief formation still behaves as a Bayesian updater.

%%%%%%%%%%%%%%%%%%%%%%%%%%%%%%% Motivation
To motivate
the proposed framework, sections \ref{CBP}, \ref{choice_functions}, \ref{EPR_section}, and \ref{statedep} present applications and extensions of our main representation results. We show that the weighted averaging axiom is closely related to many known axioms in different topics, from the Pareto axiom in Social Choice Theory to the path independence axiom in Choice Theory.

\section{Model aggregation}\label{secmodag}

Let $X$ be a nonempty set\footnote{We make no assumptions about the cardinality or topology of $X$}.
 Write $X^*$ for the set of all nonempty finite subsets of $X$. One may interpret $X$ as a (possibly infinite) set of models,  $A\in X^*$ as a finite set of models, and $X^*$ as the set of all nonempty finite sets of models. Let $H$ be a separable Hilbert space with $H=\mathbb{R}^n$ as a prototypical example.

\begin{definition} An {\em \textbf{aggregation rule}} on $X$ is a function $f:X^*\rightarrow H$, that associates with every  $A\in X^*$ a vector $f(A)\in H$.\footnote{All discussions of this paper continue to hold if $H$ is replaced by any general (possibly infinite dimensional) normed vector space.}
\end{definition}

For $x\in X$ and $A\in X^*$, one may interpret $f(\{x\})$ as the output of the model $x$, and $f(A)$ as the output of the aggregation of models contained in $A$.
The purpose of this section is to characterize aggregation rules satisfying the weighted averaging axiom/property defined below.

\begin{definition}We say that an aggregation rule $f$ satisfies the \emph{\textbf{weighted averaging}} axiom/property if for all
 $A,B\in X^*$ such that $A\cap B=\emptyset$, it holds true that
\begin{equation}\label{eqhejhdgeydg}
f(A\cup B)=\lambda f(A)+(1-\lambda)f(B)
\end{equation}
for some $\lambda \in [0,1]$ (which may depend on $A$ and $B$). We say that $f$ satisfies the \emph{\textbf{strict weighted averaging}} axiom/property if
\eqref{eqhejhdgeydg} holds true for $\lambda \in (0,1)$. We say that $f$ satisfies the \emph{\textbf{extreme weighted averaging}} axiom/property if
\eqref{eqhejhdgeydg} holds true for $\lambda \in \{0,1\}$.
\end{definition}

Two simple examples of aggregation rules satisfying the weighted averaging property are as follows.

\begin{example}\label{exampleweightedaveraging2}
Write $\mathbb{R_{++}}$ for the set of strictly positive real numbers and let  $w:X\to \mathbb{R_{++}}$ be a weight function on $X$.
For $x\in X$, let\footnote{Abusing notations we write $f(x)$ for $f(\{x\})$ for $x\in X$.} $f(x)$ be the output of model $x$.
For $A\in X^*$, define
\begin{equation}
f(A)=\sum\limits_{x\in A} \left(\frac{w(x)}{\sum\limits_{y\in A} w(y)}f(x) \right)\,.
\end{equation}
Then $f$  satisfies the strict weighted averaging property.
\end{example}

\begin{example}\label{exampleweightedaveraging1}
Consider a complete strict order $\succ$ on  $X$. Given any feature $x\in X$, let $f(x)\in H$ be the output of the model $x$.  For $A\in X^*$, write $M(A,\succ)$ for the highest order element\footnote{$M(A,\succ)\in A$ and $M(A,\succ)\succ x$ for $x\in A\setminus \{\M(A,\succ)\}$.} in $A$. For $A\in X^*$, define
\begin{equation}
 f(A)=f(M(A,\succ)).
 \end{equation}
Then $f$  satisfies the extreme weighted averaging property.
\end{example}

We will now show that all aggregation rules satisfying the strict weighted averaging property must be as in Example \ref{exampleweightedaveraging2} if $X$ contains at least three elements $x,y,z$ such $f(x), f(y)$ and $f(z)$ are not collinear.
\begin{definition}\label{defjlkjhgewjhg} An aggregation rule $f:X^*\rightarrow H$ is $\emph{\textbf{rich}}$ if the range of $f$ is not a subset of a line.
\end{definition}

\begin{theorem}\label{main_theorem}
Let the aggregation rule $f:X^*\to H$ be rich. The following are equivalent:
\begin{enumerate}
\item The aggregation rule $f$ satisfies the strict weighted averaging property.
\item There exists a weight function $w:X\to R_{++}$ such that for every $A\in X^*:$
\begin{equation}\label{main_equation}
f(A)=\frac{\sum\limits_{x\in A}w(x)f(x)}{\sum\limits_{x\in A} w(x)}.
\end{equation}
\end{enumerate}
Moreover, the function $w$ is unique up to multiplication by a positive number.
\end{theorem}

We will now show that all aggregation rules satisfying the  weighted averaging property must be of the form of a combination of  Example \ref{exampleweightedaveraging2} and \ref{exampleweightedaveraging1} if for all $x\in X$ we can find $y,z\in X$ such that  $f(\{x\}),f(\{y\})$, and $f(\{z\})$ are not collinear and the pairwise aggregation of $x,y,z$ does not satisfy the extreme aggregation property.

\begin{definition} An aggregation rule $f:X^*\rightarrow H$ is $\emph{\textbf{strongly rich}}$ if for any $x\in X$ there exist $y,z\in X$ such that:
\begin{enumerate}
\item $f(\{x\}),f(\{y\})$, and $f(\{z\})$ are not on the same line.
\item $f(\{x,y\})\notin\{f(x),f(y)\}$ and $f(\{x,z\})\notin\{f(x),f(z)\}$\footnote{In the proof of our main result, we show that as long as $f(\{x\}),f(\{y\})$, and $f(\{z\})$ are not on the same line, then $f(\{x,y\})\notin\{f(x),f(y)\}$ and $f(\{x,z\})\notin\{f(x),f(z)\}\Rightarrow f(\{y,z\})\notin\{f(y),f(z)\}$.}.
\end{enumerate}
\end{definition}

\begin{definition}
A binary relation $\succcurlyeq$ on $X$ is a \textbf{\emph{weak order}} on $X$, if it is reflexive ($x \succcurlyeq x$), transitive ($x \succcurlyeq y$ and $y \succcurlyeq z$ imply $x \succcurlyeq z$), and complete (for all $x,y \in X$, $x\succcurlyeq y$ or $y \succcurlyeq x$).
We say that $x$ is equivalent to $y$, and write $x\sim y$, if $x\succcurlyeq y$ and $y \succcurlyeq x$.
\end{definition}

\begin{theorem}\label{general_thm}
Let the aggregation rule $f:X^*\to H$ be strongly rich. Then the following are equivalent:\begin{enumerate}
\item The aggregation rule $f$ satisfies the weighted averaging axiom.
\item There exist a unique weak order $\succcurlyeq$ on $X$ and a weight function $w:X\to R_{++}$ such that for every $A\in X^*$:
\begin{equation}\label{main_equation2}
f(A)=\sum\limits_{x\in M(A,\succcurlyeq)} \left(\frac{w(x)}{\sum\limits_{y\in M(A,\succcurlyeq)}w(y)}\right)f(x) .\end{equation}
\end{enumerate}
Moreover in this case, the function $w$ is unique up to multiplication by a positive number in each of the equivalence classes of the weak order $\succcurlyeq$.
\end{theorem}

The representation \eqref{main_equation2} has two components: one is captured by the weak order $\succcurlyeq$; the other is the weight function $w$.
The weak order partitions the set of features into equivalence classes and ranks them from top to bottom. If all models $x\in A$ have the same ranking, then the outcome $f(A)$ of $A\in X^*$ is the weighted average of the outcomes of each element  $x\in A$. However, if some elements have a higher ranking than others, then the aggregation rule will ignore the lower-ordered elements.
Hence, the assessment of the aggregation rule has two steps. First, it only considers the highest-ordered elements. Then, it uses the weight function and finds the weighted average among the highest-ordered elements.\\

The richness condition is necessary for both theorems \ref{main_theorem} and \ref{general_thm}. Example \ref{examplerichness} shows that without this condition,   aggregation rules may satisfy the strict weighted averaging axiom without having a weighted average representation.

\begin{example}\label{examplerichness}
Let $X=\{x,y,z\}$ with $f(\{x\})=0, f(\{y\})=1/2, f(\{z\})=1, f(\{x,y\})=1/4, f(\{y,z\})=3/4, f(\{x,z\})=3/8$, and $f(\{x,y,z\})=7/16$. Assume that there exists a positive weight function on  $X$, and the aggregation rule over any coalition of $X$ has a representation as a weighted average of its elements. Assume that $w: X\rightarrow \mathbb{R}_{++}$ is the corresponding weight function. In order to have such a representation, we should have $f(\{x,y\})=\frac{w(x)f(x)+w(y)f(y)}{w(x)+w(y)}$. By considering the value of $f(\{x,y\}), f(\{x\})$, and $f(\{y\})$, we get $\frac{w(x)}{w(y)}=1$. Similarly, by considering the coalition $\{y,z\}$ we get $\frac{w(y)}{w(z)}=1$. By combining these two observations, we get $\frac{w(x)}{w(z)}=1$. However, considering the coalition $\{x,z\}$, and the representation $f(\{x,z\})=\frac{w(x)f(x)+w(z)f(z)}{w(x)+w(z)}$, we get $\frac{w(x)}{w(z)}=5/3$, which is a contradiction. Hence, the representation does not work in this case.
\end{example}

Assume $X$ is a subset of a normed vector space. We will now show that weight $w$ in the representation  \eqref{main_equation2} must be continuous if $f$ is continuous as defined below.
\begin{definition} An aggregation rule $f: X^*\rightarrow H$ is \emph{\textbf{(continuous)}}  if, for any  $A\in X^*\cup\{\emptyset\}$,  any  $x\in X\setminus A$, and any sequence $(x_n)_{n=1}^{\infty}\in X$,  $x_n\to x$ implies\footnote{The convergence in $X$ is with respect to the norm on $X$, and the convergence in the range of the aggregation rule is with respect to  norm of $H$.} $f(A\cup\{x_n\})\to f(A\cup\{x\})$.
\end{definition}

\begin{theorem} \label{proposition_contin}
Let X be a subset of a normed vector space and let $f:X^*\to H$ be a strongly rich continuous aggregation rule satisfying the weighted averaging property. Then the representation \eqref{main_equation2} holds true with  a continuous weight function $w:X\to R_{++}$. Furthermore, for any $x\in X$ there exists $\epsilon >0$ such that $y\sim x$ for all where $y\in X$ such that $|y-x|<\epsilon $.
\end{theorem}

We will now show that if $X$ is a convex subset of a normed vector space, then any continuous aggregation rule on $X$ under the weighted averaging axiom can only have a single equivalence class, and as a consequence, both the weighted averaging and strict weighted averaging properties lead to the representation \eqref{main_equation} for $f$.

\begin{theorem}\label{single_indiff}
Let $X$ be a convex subset of a normed vector space, and $f:X^*\to H$ a rich continuous aggregation rule satisfying the weighted averaging property.
Then, there exists a continuous weight function $w:X\to R_{++}$ such that the representation \eqref{main_equation} holds true.
\end{theorem}

\section{Preference aggregation and duality}\label{dual_Thm}
In many cases, where the range of the aggregation rule is the set of linear functionals, a simple geometrical interpretation of the weighted averaging axiom results in a related but different consistency axiom.
Let $H$ be a Hilbert space and write $\langle \cdot, \cdot \rangle$ for the associated inner product.
Let $S\subset H$ be a convex subset of $H$.
Every $h\in H$ induces a weak order (reflexive, transitive, and complete binary relation) $\succsim_h$ over the set $S$ by:
\begin{equation}\label{eqlkkjkdehdjd}
s_1\succsim_h s_2 \Leftrightarrow \langle h ,s_1\rangle \geq  \langle h,s_2\rangle
\end{equation}
Let $X$ be a non empty set. In this section we define an aggregation rule\footnote{By the Riesz representation theorem, $f$ can also be defined as function mapping $X^*$ to the space of continuous linear functionals on $H$, in which case for $A\in X^*$, $f(A)$ is identified with the unique element $h\in H$ such that $f(A)(x)=\langle h,x \rangle$ for $x\in H$.} as a function $f$ mapping $X^*$ to $H$.
Since we may interpret each $h\in H$ as a linear ranking of the elements of the set $S$, the goal of an aggregator $f$ is to attach an aggregated linear ranking to every finite subset $A$ of $X$.

\begin{example}\label{exkjedkhed}
A simple example of interpretation of $X,S$ and $f$ is as follows.
Let $X$ be a set of experts and $S$ a set of alternatives (models, decisions, choices). An expert $x\in X$ defines a ranking/preference $f(\{x\})$ over $S$. An aggregation rule $f$ is a voting mechanism enabling the aggregation of the preferences of a finite set of experts.
A rational notion of consistency (employed here and formally introduced below in Definition \ref{def_Consistent_rule}) is that if  $A, B\in X^*$ are two disjoint sets of experts such that both ($f(A)$ and $f(B)$) prefer $s_1\in S$ over $s_2\in S$, then their aggregate $f(A\cup B)$ must also prefer $s_1$ over $s_2$.
\end{example}

Observing that the order \eqref{eqlkkjkdehdjd} is invariant under scaling of $h$ we will restrict the range of aggregation rules to
the set $N_\nu=\{h\in H|\  \langle h,\nu \rangle=1\}$ for some $\nu \in H$. This restriction also avoids entirely opposite ranking directions by imposing a shared rank on $\nu$. The condition of existence of such a $\nu$ is what we call a \textbf{\textit{minimal agreement}} condition.

\begin{definition}\label{def_Consistent_rule} An  aggregation rule $f: X^*\to N_v$ is \textbf{\emph{weakly consistent}} if for all disjoint sets  $A, B\in X^*$, and for all  $s_1,s_2\in S$,
\begin{equation}\label{def_first_expr_const}
s_1 \succsim_{f(A)} s_2\ , \ s_1\succsim_{f(B)}  s_2 \Rightarrow s_1 \succsim_{f(A\cup B)}  s_2
\end{equation}
Moreover, it is \textbf{\emph{consistent}} if it also satisfies the following condition:
\begin{equation}\label{def_second_expr_const}
s_1 \succ_{f(A)} s_2\ , \  s_1 \succsim_{f(B)}  s_2 \Rightarrow s_1 \succ_{f(A\cup B)}  s_2
\end{equation}
\end{definition}

A simple duality argument (Farkas's lemma) results in the following theorem.

\begin{theorem}\label{thm_consistencyeqwa}
Let $f:X^*\to N_\nu$ be an aggregation rule. Then, the followings are equivalent:
\begin{enumerate}
\item $f$ is consistent.
\item $f$ satisfies the strict weighted averaging property.
\end{enumerate}
Moreover, the followings are also equivalent:
\begin{enumerate}
\item $f$ is weakly consistent.
\item $f$ satisfies the weighted averaging property.
\end{enumerate}
\end{theorem}

Using  Theorem \ref{main_theorem}, we immediately attain the representation of the consistent aggregation rules.%\todo{richness must be defined, }

\begin{corollary}\label{cor_AAR}
Let $f:X^*\to N_\nu$ be a consistent  rich aggregation rule. Then, there exists a weight function $w: X\to \mathbb{R}_{++}$ such that for every set of features $A\in X^*$,
\begin{equation}
f(A)=\sum\limits_{x\in A}\left(\frac{w(x)}{\sum\limits_{y\in A} w(y)}\right)f(x).
\end{equation}
Moreover, the weight function is unique up to multiplication by a positive number.
\end{corollary}

Note that we can generalize the result to the case of weakly consistent rules. 

\begin{remark}
The notion of consistency obtained as a dual interpretation of the weighted averaging is the same axiom as {\emph{\textbf{Extended Pareto}}} introduced by \cite{shapley2}. Similarly, it is the {\emph{\textbf{Combination axiom}}} in \cite{gilboa_inductive, gilboa_book}.
\end{remark}

\section{Belief Formation}\label{CBP}

In this section, we interpret the set of features as signals. Each signal contains some information about the distribution of states of nature. The role of an aggregation rule is an agent who makes a prediction about the true state of nature-based on observing some signals. In this context, the range of an aggregation rule is that of probability distributions over the states of nature. Following \cite{billot1}, an aggregation rule is a \emph{belief formation process} that associates with each finite set of signals, a \emph{belief} over the states of nature.

The representation of the belief formation process under the weighted averaging axiom is a straightforward application of the main results. Using our representation, on the one hand, we propose an extension, where the timing of signals may be important. We consider the case where an agent can receive signals in different time zones in the past. The agent tries to form a prediction at the present time, and it may perceive signals closer to the time of the prediction as more credible. To capture the representation, we introduce the \emph{stationarity} axiom, in which a belief induced by a set of received signals and their timing is the same as the belief induced by shifting the timings of all signals by a constant number to the past. 

Under stationarity, any belief formation process satisfying the strict weighted averaging axiom has a weight function over the set of signals and an exponential discount factor over each time zone. The belief associated with a set of received signals is the discounted weighted average of the beliefs associated with each signal. In this case, the weight function captures the time-independent value of each signal.

On the other hand, we interpret the set of signals as the information structure of an agent who wants to predict the true state. We interpret each subset of signals as an event in her information structure. We show that as long as the information structure has a finite cardinality, the strict weighted averaging axiom is the necessary and sufficient condition for a rich belief formation process to appear as a Bayesian updater. This result answers the question in \cite{yariv}, regarding finding a necessary and sufficient condition for a belief formation process to act as a \emph{Bayesian updating rule}.

%%%%%%%%% Average choice functions

\subsection{Belief Formation Processes}\label{bfp}

Let $\Omega=\{1,2, \ldots,n\}$ be a set of states of nature and let   $\Delta(\Omega)$ be the set of all probability distributions over $\Omega$.
 We interpret the elements of the set $X$ as disjoint signals or events.
The role of an aggregation rule over a finite subset of $X^*$ is to predict the true state of nature by assigning probabilities to each state of $\Omega$. Therefore, following  \cite{billot1}, the aggregation rules can be interpreted as a \emph{belief formation process}, which assigns a \emph{belief} to the set of states of nature after observing a finite subset of signals.

\begin{definition}A \textbf{\emph{belief formation process}} is a function $f:X^*\to \Delta(\Omega)$, that associates with every  finite set of signals $A\in X^*$, a \textbf{\emph{belief}} $f(A)\in \Delta(\Omega)$ on the states of nature.
\end{definition}

Theorem \ref{general_thm} shows that if the belief induced by the union of two disjoint finite sets of signals is on the line segment connecting the beliefs induced by each set of signals separately, then, under the strong richness condition, there exists a strictly positive weight function and a weak order over the set of signals such that the belief over any finite subset of signals is a weighted average of the beliefs induced by each of the highest-ordered signals of that subset.

By enforcing the belief formation process to use both of the induced beliefs, \textit{i.e.}, the belief induced by the union of two disjoint finite sets of signals is on the ``interior'' of the line segment connecting the induced belief of each set of signals separately, we can use Theorem \ref{main_theorem} to find the representation.
Formally, we have:
\begin{corollary}\label{CBFT_corol}
Let $f:X^*\to \Delta({\Omega})$ be a  strongly rich  belief formation process  satisfying the weighted averaging property.  Then, there exist a unique weak order $\succcurlyeq$ on $X$ and a weight function $w:X\to R_{++}$ such that for every $A\in X^*$:
\begin{equation}
f(A)=\sum\limits_{x\in M(A,\succcurlyeq)}\left(\frac{w(x)}{\sum\limits_{x\in M(A,\succcurlyeq)} w(x)}\right)f(x)
.\end{equation}
Moreover, if  $f$ satisfies the strict weighted averaging property, then the weak order $\succcurlyeq$ has only one equivalence class and for every $A\in X^*$:
\begin{equation}\label{eqlkjehde8d7h}
f(A)=\sum\limits_{x\in A}\left(\frac{w(x)}{\sum\limits_{x\in A} w(x)}\right)f(x).
\end{equation}
\end{corollary}

Although representation \eqref{eqlkjehde8d7h} is, under the strict weighted averaging property,  similar to the one in \cite{billot1}, their belief formation process is defined over \emph{sequences} of signals, in which each sequence can have multiple copies of the same signal. In contrast, we define the belief formation process over \emph{sets} of signals, and there can be only one copy of a signal in each set. Billot et al.'s main axiom, concatenation axiom, is defined over any two sequences of signals, and counts the number of each signal in each sequence. However, our strict weighted averaging property expressed for $f(A\cup B)$ does not allow   $A$ and $B$ to both contain the same signal.

%%%%%%%%%%%%%%%%%%%%    On Belief formation process

\subsection{Role of Timing}\label{roft}
We now explore the role of the timing of signals by associating signals with time labels. In that setting, a signal closer to the time of the prediction may be perceived as more credible (have more weight) compared to the same signal if it was received further in the past.
Formally, let X be the set of signals. The present time denoted by $0$, and time $t\in \mathbb{N}$ represents $t$ units of time before the present time. For a given finite subset of signals $A\in X^*$, let a function $T_A:A\to \mathbb{N}$, represent the timing of each signal in the set $A$, \textit{i.e.}, for any signal $x\in A$, $T_A(x)$ is the time of receiving the signal $x$. Given a $c\in \mathbb{N}$, $T_A+c$ represents a \emph{time shift} of size $c$ over the timing $T_A$ of a set of received signals $A$. Finally, the set $X^T=\{(A,T_A)\ |\ A\in X^*,\  T_A:A\to \mathbb{N}\}$ represents all possible realizations of the received signals. In this context, a belief formation process is a function $f:X^T\to \Delta(\Omega)$.

Our main consistency property, in addition to the strict weighted averaging property, is the \emph{stationarity} property. A belief formation process is stationary if a belief induced by a set of received signals and their timing is the same as the belief induced by a constant shift of timings of the same received signals. More precisely:
\begin{definition}A \textbf{\emph{stationary}}  belief formation is a function $f:X^T\to \Delta(\Omega)$ such that
\[f((A,T_A+c))=f(A,T_A),\]
for $A\in X^*$, $T_A:A\to \mathbb{N}$, and $c\in \mathbb{N}$.
\end{definition}
The next proposition characterizes stationarity belief-formation processes satisfying the strict weighted averaging property.
\begin{proposition}\label{time_stationary}
Let a rich and stationary belief formation process $f:X^T\to \Delta(\Omega)$ satisfy the strict weighted averaging property. Then, there exist a unique discount factor $q\in (0,\infty)$ and a unique (up to multiplication by a positive number) weight function $w:X\to \mathbb{R}_{++}$, such that for all $(A,T_A)\in X^T$:
\begin{equation}
f(A,T_A)=\frac{\sum\limits_{x\in A}q^{T_A(x)}w(x)f(x)}{\sum\limits_{x\in A} q^{T_A(x)}w(x)}.
\end{equation}
\end{proposition}

As a consequence of the representation, under the assumption of the proposition, the weight over a received signal $x\in A$ can be separated into two separate factors. One is the intrinsic value of the signal, captured by $w(x)$. The other one is the role of timing, captured by $q^{T_A(x)}$. Moreover, the only discounting that captures the role of the timing is the exponential form.
If $q=1$, the timing is not important. Hence, the belief formation process only considers the intrinsic value of each signal. However, when $q\neq 1$, the belief formation process places relatively more ($q\in (0,1)$) or less ($q\in (1,\infty)$) weight on a signal received closer to the time of the prediction.

%%%%%%%%%%%%%%%%%%%%    Bayesian Updating Rules

\subsection{Bayesian Updating }\label{bays-updating}

Let $(X,X^*\cup\{\emptyset\})$ be the measure space of events, where $X$ has a finite number of disjoint events. The space of events captures the information structure of the belief formation process. Similarly, by considering the set $\Omega=\{1,\ldots,n\}$, we denote  $(\Omega,2^{\Omega})$ as the measure space of states of nature, where $2^{\Omega}$ is the set of subsets of the set $\Omega$. For any probability distribution $d\in \Delta(\Omega)$ and any subset of the state of nature $B\in \Omega$, let $d(B)$ denote the probability of $B$ which is induced by the distribution $d$. Hence, $d(B)=\sum_{\omega\in B}d(\omega)$.

\begin{definition} A belief formation process $f: X^*\to \Delta(\Omega)$ is \textbf{\emph{Bayesian}}, if there exists a probability measure $P$ on the space $(\Omega\times X,2^{\Omega\times X})$, such that for every $A\in X^*$ and $B\in 2^{\Omega}$ we have:
\begin{equation}\label{baysian_equ}
\big(f(A)\big)(B)=\frac{P(B\times A)}{P_X(A)}
\end{equation}
where, $P_X$ is the marginal probability distribution of $P$ over $X$.
\end{definition}
 The right-hand side of the previous equation is the conditional probability of $B$ given $A$. Therefore, a Bayesian belief formation process $f$ behaves as a Bayesian updater: by observing an event $A$ in her information structure $X^*$, her prediction about the probability of the true state being in a subset $B\in\Omega$ comes from the Bayes rule. To put it differently, $\big(f(A)\big)(B)$ is equal to the conditional probability $P(B|A)$.

Our next proposition shows that our strict weighted averaging axiom is the necessary and sufficient condition for a rich belief formation process to be Bayesian.

\begin{proposition}\label{bayesupdating} A rich belief formation process is Bayesian if and only if it satisfies the strict weighted averaging property.
\end{proposition}

Note that the richness condition is crucial. Otherwise, as shown in Example \ref{examplerichness}, there are cases where a belief formation process satisfies the strict weighted averaging axiom, but it is not a Bayesian updater.
We will now present a more general version of Proposition \ref{bayesupdating}  by adding the strong richness condition and weakening the strict weighted averaging condition to the weighted averaging property. In the more general version, it is possible to have zero probability events. The belief formation process behaves as a Bayesian updater, even conditional on observing a zero probability event. To capture the idea, we need the following definition.

\begin{definition}\label{defbayup}
A class of functions $\{P_A| \ P_A:2^\Omega\times X^* \to [0,1], \ A\in X^* \}$ is a \emph{\textbf{conditional probability system}} if it satisfies the following properties:
\begin{enumerate}
    \item For every $A\in X^* $ such that $A\neq\emptyset$, $P_A$ is a probability measure on $\Omega\times X$ with $P_A(\Omega\times A)=1$.
    \item For every disjoint events $A_1,A_2\in X^*$ and for every $C\in \Omega\times X$, we have:
    $$P_{A_1\cup A_2}(C)=P_{A_1\cup A_2}(\Omega\times A_1)P_{A_1}(C)+P_{A_1\cup A_2}(\Omega\times A_2)P_{A_2}(C)$$
\end{enumerate}
\end{definition}

In this definition, the probability measure  $P_\Omega$ represents a prior probability measure, and $P_A$ represents a posterior
(conditional) probability probability given the event $A$. Therefore, for any set $B\in \Omega$, $P_A(B\times A)$ is the conditional probability of $B$ given $A$. Moreover, for any two events $A_2\subset A_1$ in $X^*$, $P_{A_1}(\Omega \times A_2)$ is the conditional probability of event $A_2$ given $A_1$.
The first property of Definition \ref{defbayup} requires that the support of the posterior probability conditioned on an event $A$ is contained in $A$.
The second property requires that,  conditional on the event $A_1\cup A_2$, the Bayes updating rule should be satisfied even if the prior probability of  $A_1\cup A_2$ is zero.

\begin{definition} A belief formation process $f: X^*\to \Delta(\Omega)$ is \textbf{\emph{rationalizable by a conditional probability system}}
$\{P_A| \ P_A:2^\Omega\times X^* \to [0,1], \ A\in X^* \}$ if every $A\in X^*$ and $B\in 2^{\Omega}$ we have:
\begin{equation}\label{baysian_equ2b}
\big(f(A)\big)(B)=P_{A}(B \times A)\,.
\end{equation}
\end{definition}

By adding the strong richness condition, the next theorem shows that the weighted averaging axiom is the necessary and sufficient condition for rationalizing a belief formation process by a conditional probability system.

\begin{proposition}\label{generalbayes}
 A strongly rich belief formation process is rationalizable by a conditional probability system if and only if it satisfies the weighted averaging axiom.
\end{proposition}

\begin{remark} \cite{yariv} considers the problem of characterizing the \emph{updating rules} (in our context the belief formation processes) that appear to be Bayesian. By providing an example, they show that their \emph{soundness condition}, our strict weighted averaging condition, is not a sufficient condition for an updating rule to behave as a Bayesian updater. However, we show that the strict weighted averaging condition is the necessary and sufficient condition as long as the belief formation process satisfies our richness condition.
\end{remark}

%%%%%%%%%%%%%%%%%%%%    Average Choice Functions
%%%%%%%%%%%%%%%%%%%%
%%%%%%%%%%%%%%%%%%%%

\section{Average Choice Functions}\label{choice_functions}

In this section, the set of features is a subset of $\mathbb{R}^n$. We interpret each feature as a choice object. The interpretation of the aggregation rule is a decision-maker that selects a choice randomly from a menu of choice objects. We model the decision-maker as an \textit{average choice function} that associates with any menu of choice objects, an average choice (mean of the distribution of choices) in the convex combination of choice objects. The average choice is easier to report and obtain rather than the entire distribution\footnote{Check \cite{ahn} for the complete discussion on merits of average choice.}. However, except for the case where elements of a menu are affinely independent, average choice does not uniquely reveal the underlying distribution of choices. 

First, using our main representation, we show that it is possible to uniquely extract the underlying distribution of choices as long as the average choice function satisfies the weighted averaging axiom.

Then, we illustrate two applications of the result. In one application, we consider the class of average choice functions that can be rationalized by a \emph{Luce rule}, \textit{i.e.}, a stochastic choice function that satisfies the \emph{independence of irrelevant alternatives} axiom (IIA) proposed by \cite{luce1}. We show that the average choice functions satisfying the strict weighted averaging axiom are exactly the ones that can be rationalized by a Luce rule. More generally, we show that the class of average choice functions satisfying the weighted averaging axiom is the same as the class of average choice functions rationalizable by a \emph{two-stage Luce} model proposed by \cite{generallucemodel}.

In the second application, we consider continuous average choice functions. First, we show that any continuous average choice function under the weighted averaging axiom is rationalizable by a Luce rule. This means that there is no continuous average choice function that is rationalizable by a two-stage Luce rule but not with a Luce rule.

Then, we illustrate a connection of our result with the one by \cite{kalai}, regarding the impossibility of an average choice function to satisfy both the \emph{path independence} axiom and continuity. 

%%%%%%%%%%%%%%%%%%%%%%%

%%%%%%%  Extended Pareto

%%%%%%%%%%%%%%%%%%%%    Primitives

\subsection{Set up}\label{choice_premitives}

In this section, $X$ is a nonempty subset of $\mathbb{R}^n$, which is not a subset of a line. For any $A\subseteq \mathbb{R}^n$, we denote by ${\rm Conv}(A)$ the set of all convex combinations of vectors in $A$.
\begin{definition} An aggregation rule $f: X^*\to \mathbb{R}^n$ is called an \textbf{\emph{average choice function}}, if for any (menu of choices)  $A\in X^*$, $f(A)\in \text{Conv}(A)$.
\end{definition}

One  of the goals of this section is to present a  connection between our weighted averaging condition and the Path Independent, Luce, and two-stage Luce choice models. The following is a corollary of  theorems \ref{general_thm} and \ref{single_indiff},

\begin{corollary}\label{choice_corol}
Let an average choice function $f:X^*\to Conv(X)$ be strongly rich. The following statements are equivalent:\begin{enumerate}
\item The average choice function $f$ satisfies the weighted averaging condition.
\item There exists a unique weak order $\succcurlyeq$ on $X$ and a unique weight function $w:X\to R_{++}$, up to multiplication over equivalence classes of the weak order  such that for every $A\in X^*$:
\begin{equation}
f(A)=\frac{\sum\limits_{x\in M(A,\succcurlyeq)}w(x)x}{\sum\limits_{x\in M(A,\succcurlyeq)} w(x)}=\sum\limits_{x\in M(A,\succcurlyeq)}\left(\frac{w(x)}{\sum\limits_{x\in M(A,\succcurlyeq)} w(x)}\right)x.
\end{equation}
\end{enumerate}
Moreover, if the average choice function $f$ satisfies continuity and the weighted averaging condition, the weight function $w$ is continuous and the weak order $\succcurlyeq $ is the equivalence order. In this case, for every $A\in X^*$:

\begin{equation}
f(A)=\frac{\sum\limits_{x\in A}w(x)x}{\sum\limits_{x \in  A} w(x)}=\sum\limits_{x\in A}\left(\frac{w(x)}{\sum\limits_{x \in  A} w(x)}\right)x.
\end{equation}

\end{corollary}

%%%%%%%%%%%%%%%%%%%%    Luce Rationalizable choice

\subsection{Luce Rationalizable Average Choice Functions}\label{average}

The following definitions are standard definitions in the context of individual decision-making.

\begin{definition}
A \textbf{\emph{stochastic choice}} is a function $ \rho : X^* \to \Delta(X)$, such that $\rho(A)\in \Delta(A)$ for any $A\in X^*$.
\end{definition}
For an average choice function $f:X^*\to Conv(X)$ and a menu $A\in X^*$, $f(A)\in \text{Conv}(A)$. Therefore, there exists a stochastic choice $ \rho : X^* \to \Delta(X)$ (which may not be unique) that rationalizes the average choice function $f$, \textit{i.e.},  $f(A)=\sum_{x\in A}\rho(x,A)x$, where $\rho(x,A)$ is the probability of selecting the element $x$ from the menu $A$.

One appealing form of a stochastic choice function is the one that satisfies \emph{Luce's IIA}, \textit{i.e.}, the probability of selecting an element over another element is independent of any other element. \cite{luce1} shows that stochastic choices that satisfy the IIA axiom are in the form of Luce rules.
\begin{definition}
A stochastic choice $\rho : X^*\to\Delta(X)$ is a \textbf{\emph{Luce rule}} if there is a  function $w : X\to R_{++}$, such that: $$ \rho (x,A)=\frac{w(x)}{\sum_{y\in A}w(y)} .$$ \\Furthermore, if $w$ is continuous, then $\rho$ is a continuous Luce rule.
\end{definition}

\begin{definition}
An average choice function $f$ is rationalizable by a stochastic choice $\rho$, if for all $A\in X^*$:
\[f(A)=\sum_{x\in A}\rho(x,A)x .\] \\Furthermore, if there exists a Luce rule that rationalizes the average choice function $f$, then $f$ is \emph{\textbf{Luce rationalizable}}.
\end{definition}

By considering our Theorem \ref{main_theorem} and corollary \ref{choice_corol}, a choice $f$ has a Luce form representation, \textit{i.e}, $f(A)=\sum\limits_{x\in A}(\frac{w(x)}{\sum\limits_{x \in  A} w(x)})x
$ if and only if it satisfies the strict weighted averaging condition. As a result:
\begin{corollary}
An average choice function is Luce rationalizable if and only if it satisfies the strict weighted averaging condition. Moreover, the Luce rule that rationalizes the average choice function is unique. \\ Furthermore, an average choice function is continuous Luce rationalizable if and only if it is continuous and satisfies the strict weighted averaging condition.
\end{corollary}

In the Luce model, the decision-maker selects each element of a given menu with a strictly positive probability. However, this is not a plausible assumption in many situations. The decision-maker may always select a better choice between two alternatives. We model this behavior by a two-stage Luce model. \cite{generallucemodel} introduces the two-stage Luce model. In this model, there exists a ranking order and a weight function over elements. A decision-maker choosing from a menu only selects the highest-ordered elements from the menu. The probability of the selection of each highest-ordered element is related to the weight associated with the element. Formally:

\begin{definition}
A stochastic choice $\rho : X^*\to\Delta(X)$ is a \textbf{\emph{two-stage Luce rule}} if there are a  function $w:X\to R_{++}$ and a weak order $\succcurlyeq$ over elements of $X$, such that:

\begin{equation}
   \rho (x,A) =
    \begin{cases}
      \frac{w(x)}{\sum_{y\in M(A,\succcurlyeq)}w(y)} & \text{if } x \in M(A,\succcurlyeq),\\
      0 & \text{otherwise} .
    \end{cases}
\end{equation}

\end{definition}

Given  $A\in X^*$, the decision-maker only selects the elements in $M(A,\succcurlyeq)$, that are the highest-ordered elements of $A$. She chooses each element of $M(A,\succcurlyeq)$ with a probability associated with its weight.

By considering our Theorem \ref{general_thm}, any average choice function under the weighted averaging axiom is rationalizable by a two-stage Luce rule.

\begin{corollary}
A strongly rich average choice function is two-stage Luce rationalizable if and only if it satisfies the weighted averaging axiom. Moreover, the two-stage Luce rule that rationalizes the average choice function is unique.
\end{corollary}

\begin{remark}
Under the continuity condition, using  Theorem \ref{single_indiff}implies that both the two-stage Luce model and Luce model are equivalent. The next section discusses this observation.
\end{remark}

%%%%%%%%%%%%%%%%%%%%      path independence

\subsection{Continuous Average Choice Functions}\label{Kalai_result}
In this section, we consider the class of continuous average choice functions satisfying the weighted averaging condition. First, we reinterpret our corollary \ref{choice_corol} as an impossibility result. This means that no continuous average choice function is rationalizable by a two-stage Luce model but not by a Luce model. Then, we show the connection with the impossibility result by \cite{kalai}, regarding the impossibility of a choice function satisfying both the path independence and continuity.

\cite{plott} extensively studies choice functions under the path independence axiom. Plott's notion of path independence requires a choice from the union of two disjoint menu
$A\cup B$, to be the choice between the choice from $A$ and the choice from $B$. Using this axiom, the choice from any menu can be recursively obtained by partitioning the elements of the menu into disjoint sub-menus. Then, the choice from the whole menu would be the choice from the choices of each sub-menu. In our setup, for an average choice function $f$, we have:
\begin{definition}$f$ satisfied the \textbf{\emph{(path independence)}} condition if
\[f(A\cup B)= f(\{f(A),f(B)\})\]
for all $A,B\in X^*$ such that $A\cap B\ = \ \emptyset$.
\end{definition}

The path independence condition is stronger than our weighted averaging condition. In other words, any average choice function under Plott's notion of path independence satisfies the weighted averaging condition. More precisely, given a choice function $f: X^* \to Conv(X)$ and two disjoint menus $A,B\in X^*$, under the path independence condition, $f(A\cup B)=f(\{f(A),f(B)\})$. By the definition of average choice functions, $f(\{f(A),f(B)\})\in Conv (f(A),f(B))$, which shows that the choice function $f$ satisfies the weighted averaging axiom.

Continuity is an appealing property of an average choice function. It specifies that by replacing an element of a menu with another element close to it, with respect to the norm of $X$, the average choice of the new menu is close to the average choice of the previous menu.
\cite{kalai} shows that there is no average choice function that satisfies both path independence axiom and continuity. Here, we reinterpret the result of corollary \ref{choice_corol} to show a more general result for average choice functions.

Corollary \ref{choice_corol} states that, for a strongly rich continuous average choice function $f: X^*\to Conv(X)$ satisfying the weighted averaging condition,  there exists a unique weight function $w: X\to \mathbb{R}_{++}$ such that for any $A\in X^*$:
\[f(A)=\sum\limits_{x\in A}(\frac{w(x)}{\sum\limits_{x \in  A} w(x)})x.
\]

There are two important observations regarding the representations above.

First, through discussions in Section \ref{average}, the representation shows that any continuous average choice function that is rationalizable by a two-stage Luce model is also rationalizable by a Luce model. Second, since the function $w$ is strictly positive, the average choice of any menu should be in the relative interior of the convex hull of members of the menu.

As a result, our impossibility result specifies that for an average choice function that satisfies the weighted averaging condition, it is impossible to satisfy the continuity condition and also to have a choice from a menu that is on the relative boundary of the elements of the menu. We summarize the observation in the following corollary\footnote{To see the connection between our corollary~\ref{cor:kalai} and the result in \cite{kalai}, it is enough to consider a menu with three non-collinear members. \cite[Thm.~1]{kalai} shows that the average choice of a path independent average choice function from any menu is the average choice of the average choice function from a sub-menu of two members of the menu. This shows that the average choice from a menu with three non-collinear members is on the line segment connecting two of the member of the menu. As a result, the choice should be on the relative boundary of the menu. That is why it cannot satisfy continuity.
}.

\begin{corollary}\label{cor:kalai}
If $X$ is a nonempty convex subset of a vector space that contains at least three non-collinear points, then an average choice function $f : X^*\to X$ that satisfies the weighted averaging condition cannot be both continuous and contains a menu $A\in X^*$, with $f(A)\in \partial^r(\text{Conv}(A))$.
\end{corollary}

\section{Extended Pareto Aggregation Rules}\label{EPR_section}

This section demonstrates an application of section \ref{dual_Thm} in the social choice problems. In this domain, each feature represents a preference ordering of individuals over a set of alternatives. Each preference ordering satisfies the axiom of \cite{vnm}. The role of an aggregation rule is to associate with each coalition of individuals another vN-M preference ordering over the set of alternatives. 

An appealing property of an aggregation rule, in this context,
is to satisfy the \textit{extended Pareto} axiom.  \cite{shapley2} introduced the extended Pareto. It specifies that, if two disjoint coalitions of individuals, each prefers an outcome over another outcome, then the union of the coalitions also should prefer the same outcome over the other one. Moreover, if one of them strictly prefer one outcome over the other one, then the union of the coalitions should also strictly prefer the same outcome over the other one.

First, we show that under a normalization of cardinal utilities of individuals and a minor richness condition, aggregation rules under the strict weighted averaging (weighted averaging) axiom are exactly aggregation rules under the \emph{extended Pareto} (\emph{extended weak Pareto}) axiom. 

Following the equivalence, we use our main representation result as a technical tool to pin down the representation of the extended Pareto aggregation rules. We show that the only possible extended Pareto aggregation is to have a positive weight over each individual in the society. Then, the aggregated preference ordering of a given group of individuals is the weighted sum of their preference ordering.

The representation can be considered as a multi-profile version of the theorem
by \cite{harsanyi} on Utilitarianism. Harsanyi considers a single profile of individuals and a variant of Pareto to get the Utilitarianism. However, in our approach, we partition a profile into smaller groups. Then, we aggregate the preference ordering of these smaller groups using the extended Pareto. Hence, we get the Utilitarianism through this consistent form of aggregation. As a result, in our representation, the weight associated with each individual appears in all sub-profiles that contain her.\footnote{
Similar to the discussion of \cite{weymark} regarding the debate of Sen-Harsanyi, our result is better to be interpreted as a representation rather than a justification of the utilitarianism.}

In Section \ref{relativeutilit}, we extend our result on extended Pareto aggregation rules to the class of \textit{generalized social welfare function}. Unlike our previous model, individuals may have different preference orderings. Therefore, the domain of the generalized social welfare function is a set of all different groups (with all possible sizes) of individuals, with each individual having all different possible preference orderings. Our definition of generalized social welfare function extends the standard definition used by \cite{arrow}, in which the domain is a set of fixed-length profiles of individuals.

For a technical reason, we restrict the set of vN-M preferences to those in which all of them strictly prefer one fixed lottery to another fixed one. We show that the only possible extended Pareto generalized social welfare functions are the ones that associate a positive number to each individual's preferences (unlike the previous section, in which each weight depends on both the individual and the whole profile), and it associates each coalition with the weighted sum of their cardinal utility using the weight associated to their preferences.

The important observation is that,  \emph{each positive weight in the representation is independent of the other individuals in any profiles}. The weight only depends on each individual and her own preference ordering.

Our representation above has a positive nature, compared to the claims by \cite{kalaischmeidler} and \cite{hylland} that the negative conclusion of Arrow's theorem holds even with vN-M preferences. Moreover, the representation provides an answer to the main concern of \cite{borger, borger2} regarding the correctness of the main theorem of \cite{dhillon}. 

\cite{dhillon} by considering a set of axioms, other than the ones by Arrow, provides one of the first axiomatizations of relative utilitarianism as a possibility result. However, \cite{borger} shows a counterexample to their representation.
Our representation fixes the error using our variant of the extended Pareto axiom and our restricted domain of the generalized social welfare function.

Finally, adding the anonymity and the weak IIA axiom of \cite{dhillon} gives us the relative utilitarianism as one possible choice of the weight function. However, the primary concern of our paper is to show that \emph{the weighted averaging of preferences is the only generalized social welfare function that respects extended Pareto}. The possible choices of weights are not our focus in this paper.

\subsection {Set up}\label{epr-def}

Let the set $M=\{0,1,\ldots,m\}$ and  $L=\{(p_1,\ldots,p_m)| \sum_{i=1}^mp_i\leq1, p_i\geq 0\}$. A lottery $p\in L$ associates the probability $p_i$ to the prospect $i\in M\setminus\{0\}$ and  $1-\sum_{i=1}^mp_i$ to the prospect $0$.

A \emph{vN-M preference}  over the set $L$ is a preference relation that satisfies the axioms of \cite{vnm} as defined below\footnote{If $R$ is a vN-M preference over the set $L$, then, by the vN-M theorem, there exists an affine representation of the preference $R$. For notational convenience, we normalize all affine representations to have the value $0$ over the prospect $0$.}.
\begin{definition}\label{defut}
We say that $R$ is a vN-M preference over the set $L$ if it is a weak order and if there exists a $u\in\mathbb{R}^m$, known as a utility, such that
for any $x,y \in L$, $\ x Ry$ if and only if $u\cdot x\geq u\cdot y$ where  ``$\cdot$'' represents the inner product in $\mathbb{R}^m$. Moreover,
the (unique) \emph{ray} $U=\{\alpha u|\ \alpha>0\}$ contains all normalized affine utilities that represent the vN-M preference $R$.
We write $\mathcal{R}$ for  the set of all vN-M preferences over  $L$ and  $\overline{R}$ for the strict part of the preference $R\in \mathcal{R}$.
\end{definition}

Let $X=\{1,\ldots,n\}$ represent the set of all agents and $X^*$ be the set of all finite subsets of $X$. Write  $\mathcal{R}^X$ for the X-Fold Cartesian product of $\mathcal{R}$. Every $R^X\in \mathcal{R}^X$ defines a \emph{preference profile} of the set of agents over the set of lotteries.

\begin{definition}\label{gar} A \textbf{\emph{group aggregation rule}} on X is a function $f:X^*\to \mathcal{R}$, that associates with every coalition of agents $A\in X^*$ a vN-M preference $f(A)\in \mathcal{R}$.
\end{definition}

An rational property of group aggregation rules is that whenever two disjoint coalitions, \textit{e.g.} $A,B\in X^*$, both prefer a lottery $x$ to another lottery $y$, then their union, $A\cup B$, also prefers the lottery $x$ to the lottery $y$.
\begin{definition} A group aggregation rule $f:X^*\to \mathcal{R}$ satisfies the \textbf{\emph{extended Pareto property}}  if for all disjoint coalitions of agents $A, B\in X^*$, and for all lotteries $x,y\in L$,
\begin{equation}\label{first_epr}
x\ f(A)\ y, \ x\ f(B) \ y \Rightarrow x\ f(A\cup B) \ y
\end{equation}
\begin{equation}\label{second_epr}
x\ \overline{f(A)}\ y, \ x\ f(B) \ y \Rightarrow x\ \overline{f(A\cup B)} \ y
\end{equation}
\end{definition}

Our last condition requires the existence of two lotteries in the set of lotteries, in which all agents strictly prefer one over the other.

\begin{definition} A group aggregation rule $f:X^*\to \mathcal{R}$ satisfies the \textbf{\emph{minimal agreement condition}} if there exist two lotteries $\overline{x},\underline{x}\in L$ such that for every agent $i\in X$, $\overline{x}\ \overline{f(i)} \underline{x}$.
\end{definition}

\begin{remark}\label{allcoalitinmac}
Let a group aggregation rule $f:X^*\to \mathcal{R}$ satisfy both the minimal agreement and extended Pareto axiom. Given two agents $i,j\in X$, by applying the strict part of the definition of the extended Pareto axiom, we have $\overline{x}\  \overline{f(\{i,j\})}\  \underline{x}$. Similarly, for every coalition of agents $A\in X^*$, recursively using the strict part of the extended Pareto axiom, we deduce $\overline{x}\  \overline{f(A)} \ \underline{x}$.
\end{remark}

\begin{remark}\label{direction}
 Let the vector $v\in \mathbb{R}^m$ be $\overline{x}-\underline{x}$, where $\overline{x},\underline{x}$ are the two lotteries in the definition of the minimal agreement condition. Let  $u_i\in \mathbb{R}^m$ represent the vN-M preference $f(i)$. Hence, $\overline{x}\ \overline{f(i)} \ \underline{x}$ if and only if $u_i\cdot v>0$. Therefore, the definition of the minimal agreement condition is equivalent to the existence of a direction $v\in \mathbb{R}^m$ such that for all $i\in X$, $u_i\cdot v>0$. this interpretation of $v$ is exactly the role of $\nu$ in section \ref{dual_Thm}.
\end{remark}

%%%%%%%%%%%%%%%%%%%%    The representation

\subsection {The Representation of Extended Pareto Group Aggregation Rules}\label{epr-group-rep}
%%%%%%%%%%%%%%%%%%%

%%%%%%%%%%%%%%%%%%%%
In this section, we assume that the group aggregation rule $f:X^*\to \mathcal{R}$ satisfies the minimal agreement condition. In particular, we assume that all agents strictly prefer the lottery $\overline{x} \in L$ over the lottery $\underline{x}\in L$. Considering remark \ref{direction}, we define $v=\overline{x}-\underline{x}$ as the direction that every agent agrees on. For a coalition of agents $A\in X^*$, let the ray $U_A$ represents the vN-M preference $f(A)$.
Let $H:=\{u\in\mathbb{R}^m|\ u\cdot v=1\}$ represent the normalization of utilities in which the difference of the value of utility of the lottery $\overline{x}$ and the lottery $\underline{x}$ is exactly 1. For every coalition of agents $A\in X^*$, there is a unique cardinal utility  $\hat{u}_A\in U_A$, such that $\hat{u}_A$ is in $H$.
For the rest of the section, for every coalition $A\in X^*$, we consider the unique cardinal utility $\hat{u}_A\in H$ to represent the vN-M preference $f(A)$. Using this representation, we can represent the group aggregation rule $f:X^*\to \mathcal{R}$, by a \textbf{\emph{normalized group aggregation rule}} $f_H:X^*\to \mathbb{R}^m$, where $f_H(A)=\hat{u}_A$.

\begin{remark}
Without loss of generality, we can assume that the lottery $\underline{x}$ in the definition of the minimal agreement condition is just the lottery $0$. In that case, the space $H$ is vN-M preferences with the value $0$ for the lottery $0$ and the value $1$ for the lottery $\overline{x}$.
\end{remark}

The next proposition, which is the same as Theorem \ref{thm_consistencyeqwa}, shows that under the representation of the vN-M preference $f(A)$ by the $\hat{u}_A$, the extended Pareto property is equivalent to the strict weighted averaging property. Formally, we have:
\begin{corollary}\label{EPR=SCA}
Let a group aggregation rule $f:X^*\to \mathcal{R}$ satisfy the minimal agreement condition with $v\in \mathbb{R}^m$ as the direction on which all agents agree. Then, the following are equivalent:
\begin{enumerate}
\item $f$ satisfies the extended Pareto property.
\item $f_H$ satisfies the strict weighted averaging property.
\end{enumerate}
\end{corollary}

Using the result of Theorem \ref{main_theorem}, we deduce the representation of the extended Pareto group aggregation rules.%\todo{richness must be defined, }
\begin{corollary}\label{EPR_rep}
Let a rich group aggregation rule $f:X^*\to \mathcal{R}$ satisfy both the extended Pareto property and minimal agreement condition. Then, there exists a weight function $w: X\to \mathbb{R}_{++}$ such that for every coalition of agents $A\in X^*$,
\begin{equation}
f_H(A)=\sum\limits_{i\in A}\left(\frac{w(i)}{\sum\limits_{j\in A} w(j)}\right)f_H(i).
\end{equation}
Moreover, the weight function is unique up to multiplication by a positive number.
\end{corollary}

As shown in Example \ref{examplerichness}, the richness condition is crucial. The richness here is equivalent to the existence of three non-collinear ``normalized'' cardinal utilities in the space $H$ (the range of the aggregation rule).
We can interpret the theorem as a generalization of the main theorem of \cite{harsanyi} on Utilitarianism. However, our result shows the connection between weights of individuals in different sub-coalitions of the main profile.

To see the connection with Harsanyi's result, we rewrite the theorem in an additive form: let the group aggregation rule $f:X^*\to \mathcal{R}$ satisfy both the extended Pareto property and minimal agreement condition. Then, there exists a weight function $w: X\to \mathbb{R}_{++}$ such that for every coalition of agents $A\in X^*$, $f(A)$ has the following representation:
\begin{equation}\label{epr-eq-sum}
\sum\limits_{i\in A}w(i)f_H(i).
\end{equation}
Defining $u(i):=w(i)f_H(i)$ for $i\in X$, we can rewrite equation \ref{epr-eq-sum} in the  additive form $ \sum\limits_{i\in A}u(i)$.
Moreover, if we consider only the representations with the value $0$ for the lottery $0$, this representation is unique up to multiplication by a positive number.

%\todo[inline, color=red]{
%A)indifferent preferences,

%B)richness condition,

%C)General case of the set of Agents $X$ has infinite number of agents,

%D) extended weak Pareto, and more general version of the theorem (with strong richness).}

\subsection {The Representation of Extended Pareto Generalized Social Welfare Functions }\label{relativeutilit}

The setup of this subsection is the same as the one in the previous section. Without loss of generality, we assume that the lottery $\underline{x}\in L$, in the definition of the minimal agreement, is the vector $0$. Let $\overline{x}\in L$ be any lottery other than $0$. Define $\mathcal{R}_{\overline{x}}\subset \mathcal{R}$ as the set of all vN-M preferences that strictly prefer $\overline{x}$ to $0$. Let $\mathcal{R}_{\overline{x}}^X$ be the X-fold Cartesian product of $\mathcal{R}_{\overline{x}}$. Every $R\in \mathcal{R}_{\overline{x}}^X$ defines a preference profile of the set of individuals. For any coalition $A\in X^*$ and for any preference profile $R\in \mathcal{R}_{\overline{x}}^X$, let $R_A\in \mathcal{R}_{\overline{x}}^A$ denote the restriction of the profile $R$ to the coalition $A$.

As in Definition \ref{defut}, we can represent each preference $R\in \mathcal{R}$ by a unique ray $U_R=\{\alpha u|\ \alpha>0\}$, where $u\in \mathbb{R}^m$ is a cardinal utility representing $R$. Moreover, for any preference $R\in \mathcal{R}_{\overline{x}}$, there should be a unique cardinal utility $u_R\in U_R$ with $u_R\cdot \overline{x}=1$. Write $H=\{u\in \mathbb{R}^m|\ u.\overline{x}=1\}$ for the space of all cardinal utilities attaining value $0$ at the lottery $0$ and the value $1$ at the lottery $\overline{x}$. Let the function $u_H:\mathcal{R}_{\overline{x}}\to H$ associate each preference $R\in \mathcal{R}_{\overline{x}}$ with the unique cardinal utility $u_H(R)\in H$ that represents it. This function is a bijection associating each preference to the unique cardinal utility attaining value $0$ at the lottery $0$ and value $1$ at the lottery $\overline{x}$.

Write $\mathcal{R}_X\subset \mathcal{R}_{\overline{x}}^X  $ for the set of all profiles where the representation of individuals' cardinal utilities in the space $H$ is not a subset of a single line. Formally, we define $\mathcal{R}_X=\{R\in \mathcal{R}_{\overline{x}}^X|\ d(\{u_H(R_i)|\ i\in X\})>1\}$, where $d(\{u_H(R_i)|\ i\in X\})$ is the dimension of the smallest linear variety containing all $u_H(R_i),\  i\in X$.\footnote{There should be at least four alternatives; otherwise, $\mathcal{R}_X$ is the empty set.}

Finally, write $\mathcal{R}_X^*=\{R\in \mathcal{R}_{\overline{x}}^A|\  A\subseteq X, R\in  \mathcal{R}_X  \}$ for all the profiles in $\mathcal{R}_X$ and all sub-coalitions of those profiles. $\mathcal{R}_X^*$ is the domain of our generalized social welfare functions. Formally, we have:

\begin{definition}
A \textbf{\emph{generalized social welfare function}} on $\mathcal{R}_X$ is a function $f:\mathcal{R}_X^*\to \mathcal{R}$, that associates with any coalition $A\in X^*$ and any profile $R\in \mathcal{R}_{\overline{x}}^X$ a preference $f(R_A)\in \mathcal{R}$. Moreover, we assume that for any individual $i\in X$, and any profile $R\in \mathcal{R}_{\overline{x}}^X$ , $f(R_i)=R_i$.
\end{definition}

In our setup, the domain of generalized social welfare functions is a rich set of all sizes of profiles. Moreover, it satisfies the \textbf{\emph{Individualism}} axiom, which means that it associates any individual preference to the same preference.

The connection between profiles of different sizes is the \textbf{\emph{extended Pareto property}}. The extended Pareto property states that if the associated preference ordering of two disjoint coalitions of individuals, $A$ and $B$, each prefer a lottery $x$ to $y$, then the associated preference ordering of the union of the coalition with the same preference as before should also prefer $x$ to $y$.

\begin{definition} A generalized social welfare function $f:\mathcal{R}_X^*\to \mathcal{R}$ satisfies the \textbf{\emph{extended Pareto property}} if for every preference profile $R\in \mathcal{R}_X$ and for any two disjoint coalitions $A,B\in X^* $, and for all lotteries $x,y\in L$,

\begin{equation}
x\ f(R_A)\ y, \ x\ f(R_B) \ y \Rightarrow x\ f(R_{A\cup B}) \ y
\end{equation}
\begin{equation}
x\ \overline{f(R_A)}\ y, \ x\ f(R_B) \ y \Rightarrow x\ \overline{f(R_{A\cup B})} \ y
\end{equation}

\end{definition}

Our main result of this section characterizes the class of extended Pareto generalized social welfare functions.

\begin{theorem}\label{EPR_rep_main}
Let $X$ be a set of individuals with $|X|\geq 4$. The generalized social welfare function $f:\mathcal{R}_X^*\to \mathcal{R}$ satisfies the extended Pareto property if and only if there exists a weight function $w:X\times\mathcal{R}_{\overline{x}} \to \mathbb{R}_{++}$, such that for any coalition $A\in X^*$ and any preference profile $R\in \mathcal{R}_X$, $f(R_A)$ has the following representation:

\begin{equation}\label{EPR_rep_main_eq}
u_H(f(R_A))=\sum\limits_{i\in A}\left(\frac{w(i,R_i)}{\sum\limits_{j\in A} w(j,R_j)}\right)u_H(R_i).
\end{equation}
Moreover, the weight function is unique up to multiplication by a positive number.
\end{theorem}

\begin{remark}
 We can rewrite the theorem to specify that the generalized social welfare function $f:\mathcal{R}_X^*\to \mathcal{R}$ satisfies the extended Pareto axiom if and only if there exists a weight function $w:X\times\mathcal{R}_{\overline{x}} \to \mathbb{R}_{++}$, such that for any coalition $A\subseteq X$ and any preference profile $R\in \mathcal{R}_X$, $f(R_A)$ has the following representation:
\begin{equation}\label{EPR_rep_main_eq_1}
\sum\limits_{i\in A}w(i,R_i)u_H(R_i).
\end{equation}
Note that each weight depends only on the associated individual's preferences and not on the other individuals.
\end{remark}

The weight function in the representation depends on each individual's index. However, adding the classical \textbf{\emph{Anonymity}} condition makes the weight function independent of individual's indexes.
\begin{definition}
An extended Pareto aggregation rule $f:\mathcal{R}^*_X\to \mathcal{R}$ satisfies the \textbf{\emph{Anonymity}} condition, if  any permutation of the indexes of individuals does not change the generalized social welfare function.
\end{definition}

The anonymity condition makes any extended Pareto generalized social welfare functions independent of the individual's indexes. Hence, the uniqueness of the weight function in Theorem \ref{EPR_rep_main} makes the weight function, associated with an anonymous extended Pareto aggregation rule, independent of the indexes. Therefore, we have:
\begin{corollary}
\label{EPRnutral_rep_main}
Let $X$ be a set of individuals with $|X|\geq 5$. The extended Pareto generalized social welfare function $f:\mathcal{R}_X^*\to \mathcal{R}$ satisfies the Anonymity condition if and only if there exists a weight function $w:\mathcal{R}_{\overline{x}} \to \mathbb{R}_{++}$, such that for any coalition $A\in X^*$ and any preference profile $R\in \mathcal{R}_X$, $f(R_A)$ has the  representation:
\begin{equation}\label{EPR_rep_main_eq-2}
u_H(f(R_A))=\sum\limits_{i\in A}\left(\frac{w(R_i)}{\sum\limits_{j\in A} w(R_j)}\right)u_H(R_i).
\end{equation}
Or, equivalently, if and only if $f(R_A)$ has the representation:
\begin{equation}\label{nut-epr}
\sum\limits_{i\in A}w(R_i)u_H(R_i).
\end{equation}
Moreover, the weight function is unique up to multiplication by a positive number.
\end{corollary}

The positive nature of our result appears to contradict the conjectures by \cite{kalaischmeidler} and \cite{hylland} that the negative conclusion of the impossibility theorem by \cite{arrow} holds even with vN-M preferences. However, other than the differences between our model and theirs, we only consider the restricted domain where all preferences prefer the lottery $\overline{x}$ over the lottery $\underline{x}$. As discussed before, the definition of our restricted domain is crucial in corollary \ref{EPRnutral_rep_main}.

\begin{remark}
Relative utilitarianism can be obtained by adding the weak IIA axiom of  \cite{dhillon}: the weight function normalizes each preference such that the difference between the cardinal utility of the best alternative and the worst alternative becomes $1$. In other words, for any preference $R\in \mathcal{R}_{\overline{x}}$, $w(R)=\frac{1}{\underset{j}{max} \ (u_H(R))_j- \underset{j}{min}\ (u_H(R))_j }$.
\end{remark}

%\todo[inline=yes, color=red]{More on remark 13}

%%%%%%%%%%%%%%%%%%%%    A Conditional Subjective Expected Utility Theory of State-Dependant Preferences
%%%%%%%%%%%%%%%%%%%%
%%%%%%%%%%%%%%%%%%%%

\section {A Conditional Subjective Expected Utility Theory of State-Dependant Preferences}\label{statedep}
The choice-theoretic foundation of subjective expected utility was developed by the seminal works of \cite{ramsey}, \cite{savage}, and \cite{anscombe}. In the standard model, the decision-maker has a ranking over acts (state-contingent outcomes). The representation of this ranking consists of a subjective probability over the set of states, capturing the decision maker's beliefs, and a cardinal utility representing the decision maker's tastes over the set of outcomes, independent of the realization of the true state. However, in many applications, such as models for buying health insurance, the independence of the utility and the set of states is not a plausible assumption\footnote{Check \cite{arrow2}, \cite{cookgraham}, and \cite{karniinsurance} 
 for more discussions.}.  
 
In this section, we provide a simple theory of subjective expected utility of state-dependent utility by reinterpreting our representation of extended Pareto aggregation rules. We build our model using the framework of \cite{anscombe}. In our model, the decision-maker has a preference ordering over the set of conditional constant acts. This means that given any fixed event, the decision-maker has a hypothetical conditional preference ordering over the set of lotteries, representing her conditional preference condition on learning that only that event is happening\footnote{In Section \ref{statedep}, we illustrate another interpretation of hypothetical conditional preferences by providing a preference ordering over the set of conditional constant acts. }. Each of these hypothetical conditional preferences satisfies the axioms of \cite{vnm}, which means each has an affine representation. We show that as long as the class of hypothetical conditional preferences satisfies the extended Pareto axiom, there is a subjective probability measure over the set of states and a state-dependent utility over the set of alternatives. The class of hypothetical conditional preferences has a representation in the form of conditional expectation with respect to the subjective probability and the state-dependent utility.

The result shows that the extended Pareto is the main force behind the separation of the belief and the state-dependent utility. However, the representation is not unique. Hence, the challenge is to provide meaning to a decision maker's prior beliefs when utility is state-dependent. We get the uniqueness by adding a stronger version of our minimal agreement condition. The strong minimal agreement condition specifies that there exist two lotteries where one is strictly preferred to the other, regardless of states. Moreover, the decision maker's conditional preference for each of them is independent of the set of states.

We show that under the strong minimal agreement, the belief is unique. Moreover, the state-dependent utility is unique up to affine transformation.
 
\subsection{Set up and main result}
In this section, we develop a simple theory of subjective expected utility of state-dependant utility by reinterpreting the results of sections \ref{epr-group-rep} and \ref{bays-updating}. 
Our model is built using the framework of \cite{anscombe}. Let $\Omega=\{1,2,\ldots, n\}$ be a finite set of states of nature. The finite set $M$ represents outcomes. The simplex $L=\Delta(M)$ represents the set of lotteries over the set $M$. A lottery $l\in L$ associates the probability $l_i$ to the outcome $i\in M$.
Let the $O\notin M$.
In our setup, the objects of choice are  \emph{\textbf{conditional constant acts}}. For any lottery $l\in L$, and any event $A\in 2^{\Omega}\setminus \emptyset$, the function $f:\Omega\to L\cup \{O\}$ such that $f(w)=l$ for $\omega \in A$ and  $f(w)=O$ for $\omega \in A^c$ is termed a conditional constant act and denoted by  $f=(l,A,O,A^c)$. The interpretation is that if the event $A$ is realized, the decision maker faces the lottery $l$. Otherwise, $O$ will be realized. We assume that the decision maker has a preference relation $\succcurlyeq$, not necessarily a complete relation, over the set of conditional constant acts.

Let $F=\{(l,A,O,A^c)|\ \emptyset \neq A\in 2^{\Omega} , l\in L\}$ represent the set of conditional constant acts. For any event $\emptyset \neq A\in 2^{\Omega}$, let $F_A=\{(l,A,O,A^c)| l\in L\}$ be the set of all conditional constant acts attaining a lottery on the event $A$ and staying out on the event $A^c$. We represent the conditional preference ordering of the decision maker over $F_A$ by  $\succcurlyeq_{A}$. For any two lotteries $l_1,l_2$, we write $l_1 \succcurlyeq_{A} l_2$ as a shorthand of $(l_1,A,O,A^c)\succcurlyeq (l_2,A,O,A^c)$.

Our interpretation of conditional preference ordering is related to the models  developed by \cite{lucekrantz}, \cite{fishburn}, \cite{skiadas1}, and \cite{karnifoundationofbayes}. However, there is another interpretation of the conditional preference similar to the conditional decision model of \cite{ghirardato}. In this interpretation, we assume that the decision-maker may receive some information that only $\omega \in A$ can be realized. In this case, $\succcurlyeq_{A}$ represents the decision maker's \textit{ex-post} preference over the set of lotteries. Similarly, $\succcurlyeq_{\Omega}$ represents her \textit{ex-ante} preference over exactly the same set of lotteries.

Regardless of the interpretation, the goal is to provide a theory that connects the class of conditional preferences through the Bayes updating.
Formally, our goal is to find sufficient conditions under which that there exists a state-dependent utility function $u:\Omega\times M\to \mathbb{R}$ and a subjective probability measure $P:\Omega\to \mathbb{R}_{++}$, such that for every two lotteries $x,y\in L$, and any event $A$ the following holds:

\begin{equation}\label{bayes-conditional}
    x\succcurlyeq_A y \Leftrightarrow \sum\limits_{\omega\in A}P(\omega|A)E^x[u(\omega,\cdot)]\geq
    \sum\limits_{\omega\in A}P(\omega|A)E^y[u(\omega,\cdot)] .
\end{equation}

In the equation above, $E^x[u(w,\cdot)]$ represents the expected utility of the state-dependent utility $u$ in the state $\omega$ and with respect to the lottery $x$. The right-hand side of the equation is comparing the conditional expectation utility of the lottery $x$ and $y$, with respect to the subjective probability measure $P$ and the state-dependent utility $u$.  The importance of the result is that the probability measure $P$ depends on the event $A$ through the Bayes rule. We will obtain \ref{bayes-conditional} from the following axioms/conditions.

\begin{axiomN} \emph{\textbf{(Weak Order)}} For any event $A$, the conditional preference $\succcurlyeq_A$ is complete and transitive.
\end{axiomN}

\begin{axiomN} \emph{\textbf{(vN-M Continuity)}} For any event $A$ and for every $x,y,z\in L$, if $x\succcurlyeq_A y\succcurlyeq_A z$, there exist $\alpha, \beta \in (0,1)$ such that
$$\alpha x+(1-\alpha)z\succcurlyeq_A y\succcurlyeq_A \beta x+(1-\beta)z$$
\end{axiomN}

\begin{axiomN} \emph{\textbf{(Independence)}} For any event $A$, every $x,y,z\in L$, and every $\alpha \in (0,1)$,
$$ x\succcurlyeq_A y \Rightarrow \alpha x+(1-\alpha)z\succcurlyeq_A \alpha y+(1-\alpha)z $$
\end{axiomN}

\begin{axiomN} \emph{\textbf{(extended Pareto)}} For any two disjoint events $A,B$, and for every $x,y\in L$,

\begin{equation}
x\ \succcurlyeq_A\ y, \ x\ \succcurlyeq_B \ y \Rightarrow x\ \succcurlyeq_{A\cup B} \ y
\end{equation}
\begin{equation}
x\ \succ_A\ y, \ x\ \succcurlyeq_B \ y \Rightarrow x\ \succ_{A\cup B} \ y
\end{equation}
\end{axiomN}

\begin{axiomN} \emph{\textbf{(Minimal  Agreement)}} There exist two lotteries $\overline{x},\underline{x}\in L$ such that for every  $\omega \in \Omega$, $\overline{x}\ \succ_\omega \underline{x}$.
\end{axiomN}

\begin{axiomN} \emph{\textbf{(Richness)}}
There exist three states $\omega_1,\omega_2,\omega_3\in \Omega$ such that for any $\omega \in\{\omega_1,\omega_2,\omega_3\}$, there exist two lotteries $x,y\in L$ where $x \succ_\omega y$ and $y\succcurlyeq_{\omega'} x$ for $\omega'\in \{\omega_1,\omega_2,\omega_3\}\setminus\{\omega\}$.

\end{axiomN}

By considering these six axioms, we can rationalize the behavior of the decision-maker as a subjective expected utility maximizer with a state-dependent utility.

\begin{theorem}\label{sdutil}
Suppose that the decision maker's conditional preferences satisfy axioms 6.1-6.6, then there exist a function $u:\Omega\times M\to \mathbb{R}$ and a probability measure $P:\Omega\to \mathbb{R}_{++}$, such that for every two lotteries $x,y\in L$, and any event $A$, the following holds:

\begin{equation}\label{sceusdu}
    x\succcurlyeq_A y \Leftrightarrow \sum\limits_{\omega\in A}P(\omega|A)E^x[u(\omega,\cdot)]\geq
    \sum\limits_{\omega\in A}P(\omega|A)E^y[u(\omega,\cdot)]
\end{equation}
\end{theorem}

The proof is similar to the proof of corollary \ref{EPR_rep}.
The probability measure $P$ is not unique. Let $Q: \Omega \to \mathbb{R}_{++}$ be any probability measure on $\Omega$; by defining a state-dependent utility $w(\omega,x)=\frac{u(\omega,x)}{Q(\omega)}$,  equation \ref{sceusdu} continues to hold with $Q$ and $w$. However, if we change the minimal agreement axiom to a stronger version, we attain the uniqueness. In the stronger version of the minimal agreement, we assume that the decision maker's preferences over the lotteries $\overline{x}, \underline{x}$ is indifferent to the realization of the states. Formally:

\begin{axiomN} \emph{\textbf{(Strong Minimal  Agreement)}} There exist two lotteries $\overline{x},\underline{x}\in L$ such that for every  $\omega \in \Omega$, $\overline{x}\ \succ_\omega \underline{x}$. Moreover, $(\overline{x},\{\omega_1\},O,\Omega\setminus\{\omega_1\}) \sim (\overline{x},\{\omega_2\},O,\Omega\setminus\{\omega_2\}$) and $(\underline{x},\{\omega_1\},O,\Omega\setminus\{\omega_1\}) \sim (\underline{x},\{\omega_2\},O,\Omega\setminus\{\omega_2\})$ for all $\omega_1,\omega_2\in \Omega$.
\end{axiomN}

 Conceptually, this axiom is closely related to A.0 axiom by \cite{karnifoundationofbayes}. However, unlike Karni's axiom, we do not need these two lotteries to be the best and worst lotteries in the set of lotteries. Our model only needs two lotteries, with one strictly preferred to the other, regardless of states. Moreover, the decision maker's conditional preference for each of them is independent of the set of states.

By replacing the minimal agreement axiom with the strong minimal agreement axiom, we can ``uniquely'' separate the belief from the state-dependent preference.

\begin{theorem}\label{sdutil2}
Suppose that the decision maker's conditional preferences satisfy axioms 6.1-6.7, then there exist a function $u:\Omega\times M\to \mathbb{R}$ and a probability measure $P:\Omega\to \mathbb{R}_{++}$, such that for every two lotteries $x,y\in L$, and every event $A$, the following holds:

\begin{equation}\label{sdutileq}
    x\succcurlyeq_A y \Leftrightarrow \sum\limits_{\omega\in A}P(\omega|A)E^x[u(\omega,\cdot)]\geq
    \sum\limits_{\omega\in A}P(\omega|A)E^y[u(\omega,\cdot)]
\end{equation}
Moreover, the probability measure $P$ is unique and the function $u$ is unique up to affine transformations.
\end{theorem}
\begin{proof}
Based on Theorem \ref{sdutil}, there exists a pair $(P,u)$ satisfying equation \ref{sdutileq}.
To prove the uniqueness, we assume that $(P_1,u_1)$ and $(P_2,u_2)$ both represent the same class of conditional preferences. By considering the conditional preference $\succcurlyeq_{\omega}$ and the vN-M Theorem, we know that $u_2(\omega,.)=\alpha_{\omega}u_1(\omega,.)+\beta_\omega$. By using the strong minimal agreement axiom, we have $u_1(\omega_1,\overline{x})=u_1(\omega_2,\overline{x})$, $u_2(\omega_1,\overline{x})=u_2(\omega_2,\overline{x})$, $u_1(\omega_1,\underline{x})=u_1(\omega_2,\underline{x})$, and $u_2(\omega_1,  \underline{x})=u_2(\omega_2,\underline{x})$ for any two states $\omega_1,\omega_2\in \Omega$. Therefore, $\alpha_{\omega_1}=\alpha_{\omega_2}$ and $\beta_{\omega_1}=\beta_{\omega_2}$ for all $\omega_1,\omega_2\in \Omega$. Hence, $u_2(\omega,.)=\alpha u_1(\omega,.)+\beta$ for all $\omega\in \Omega$.

We consider an event $A$. Both $(P_1,u_1)$ and $(P_2,u_2)$ represent the conditional preference $\succcurlyeq_A$. Considering the pair $(P_2,u_2)$, $\succcurlyeq_A$ has the representation  
\[\begin{split}\sum\limits_{\omega\in A}P_2(\omega|A)E^{(.)}[u_2(\omega,.)]&=\sum\limits_{\omega\in A}P_2(\omega|A) E^{(.)}[\alpha u_1(\omega,.)+\beta]\\&=\alpha \sum\limits_{\omega\in A}P_2(\omega|A)E^{(.)}[u_1(\omega,.)] +\beta\,.\end{split}\]

Since, $\alpha$ is strictly positive, the last representation is the same as\\ $\sum\limits_{\omega\in A}P_2(\omega|A)E^{(.)}[u_1(\omega,.)]$. However, using the other pair, $(P_1,u_1)$, we get the representation $\sum\limits_{\omega\in A}P_1(\omega|A)E^{(.)}[u_1(\omega,.)]$.
Therefore, for any event $A$, $\sum\limits_{\omega\in A}P_2(\omega|A)E^{(.)}[u_1(\omega,.)]$ and  $\sum\limits_{\omega\in A}P_1(\omega|A)E^{(.)}[u_1(\omega,.)]$ both represent the conditional preference $\succcurlyeq_A$.
Using the richness axiom, strong minimal agreement condition, and uniqueness of corollary \ref{EPR_rep}, we have $P_1=P_2$.
This completes the proof.
\end{proof}

\section{Related Literature}

Our methods are applicable to different areas of economic theory, and generalize existing ideas in those areas. In particular, instances of our weighted averaging axiom appear in several different papers. 

The theory of Case-Based Prediction is developed by the seminal works of  \cite{gilboa1,gilboa_inductive, gilboa_book} and \cite{billot1}. In this context, the \emph{concatenation} axiom proposed by \cite{billot1}, is closely related to the strict case of our axiom. However, there are differences between the two axioms. As discussed in detail in Section \ref{bfp}, their belief formation process is defined over ``sequences'' of cases, in which each sequence can have multiple copies of the same case. The role of the concatenation axiom is to count the number of each case. However, in our framework, we define our axiom over ``sets'' of signals, in which in each set there is only one copy of each signal. Moreover, our axiom is defined over disjoint sets. By weakening our definition for any two general sets, our result does not hold anymore.
Using the duality argument, the consistency axiom of our paper is the \emph{combination axiom} of \cite{gilboa_inductive}. However, our goal is to show that the combination axiom can be obtain by weaker version of concatenation axiom using a simple duality argument.

In the paper by \cite{yariv}, they provide an example, on binary state space, to show that their \emph{soundness} condition is not a sufficient condition for an updating rule to behave as a Bayesian rule. However, we show that under our richness assumption, the strict weighted averaging axiom (which is the same as their soundness condition) is the necessary and sufficient condition for an updating rule to behave as a Bayesian. We also generalize our result for the class of updating rules that can be rationalized by a conditional probability system.

In the context of choice theory, \cite{ahn} introduces a model of continuous average choice over convex domains. In this application, we generalized their result in many ways. First, their result holds for the strict case of our axiom. Moreover, continuity and convexity are the two important forces behind their result. However, we show that the strictness of an average choice function, continuity, or convexity of the domain are not the main forces behind extracting the underlying distribution of choices. The main force is our weighted averaging axiom. Moreover, we show that it is possible to rationalize an average choice function by a two-stage Luce model, as long as it satisfies our weighted averaging axiom.

The path independence choice functions are extensively studied by \cite{plott}.
Our representation of average choice functions under the weighted averaging axiom and continuity generalizes the results by \cite{kalai} and \cite{machina}, regarding the impossibility of a choice function under both the path independence and the continuity. 

In the context of social choice, \cite{dhillon} and \cite{shapley1} study variants of extended Pareto rules.

\cite{shapley1} study the extended Pareto rule over vN-M preferences by relaxing the completeness axiom. Besides the technical and conceptual differences between the two approaches and results, their model depends on their \emph{non-degeneracy} condition. The condition is only satisfied when there is a spanning tree over the preferences, and every three consecutive preferences in the spanning tree are linearly independent. However, the richness condition of our theorem only requires three linearly independent vectors among the whole set of preferences. Moreover, our result can be applied even for the class of extended weak Pareto aggregation rules under our strong richness condition. Note that our primary goal in this paper is to show that extended Pareto and extended weak Pareto are special cases of our weighted averaging axiom (under the minimal agreement condition).

The papers by \cite{dhillon}, \cite{dhillon2}, and \cite{borger2} each by considering different sets of axioms, other than Arrow's, provide an axiomatization of relative utilitarianism as a positive result. The paper by \cite{dhillon} is the closest one to ours. Dhillon considers a variant of extended Pareto to get a weighted averaging structure. However, \cite{borger} shows a counterexample to the representation. We restrict the domain and use our definition of extended Pareto to get the weighted averaging structure as a consequence of our main theorem. Again, the technique we developed can also be used to provide a representation of the extended weak Pareto social welfare functions.

Finally, there are many papers and different approaches to address the shortcomings of subjective expected utility theory. Note that our goal, in this context, is to explain the basic underlying structure that lets us separate beliefs and state-dependent utilities.

\cite{karnischmeidlervind} and \cite{karnischmeidler} use hypothetical preferences on hypothetical lotteries to obtain the identification of the beliefs and state-dependent preferences. \cite{derez}, \cite{derezrustichini}, and \cite{karniwithoutstate} present different theories to identify state-dependent preferences in situations where moral hazard is present. 

\cite{lucekrantz} and \cite{fishburn} use preferences on enlarged choice space of all conditional acts to model subjective expected utility of state-dependent preferences.  However, our paper only considers the hypothetical conditional preferences on the set of conditional ``constant'' acts. We find the necessary and sufficient condition that our conditional preferences are related to each other through a subjective probability and a state-dependent utility. 

The papers by  \cite{skiadas1},  \cite{ghirardato}, and \cite{karnifoundationofbayes} are closely (conceptually) related to our main result of Section \ref{statedep}. However, there are many differences between each result. Moreover, our goal is to build the model that only extended Pareto derives the separation of beliefs and state-dependent preferences.

\cite{skiadas1} presents a nonexpected utility model, by considering hypothetical preferences over the set of act-event pairs. His coherence axiom has the same role as the extended Pareto axiom in our setup. However, he used the solvability axiom to be able to apply the Debreu's additive representation theorem. In our paper, we consider the class of conditional vN-M preferences. As a result, we only require the extended Pareto for our representation.

\cite{karnifoundationofbayes} presents a general model with a preference ordering over the set of unconditional acts. Using the preference order, he defines the set of conditional preferences over the set of all conditional acts. Therefore, to connect the class of conditional preferences, the model needs the existence of the constant-valuation acts. Moreover, the cardinal and ordinal coherence axioms are the main forces behind obtaining the Bayesian updating in his representation. However, in our more restricted domain, we only need the extended Pareto to get our representation. 

Finally, \cite{ghirardato}, by replacing Savage's sure-thing principle by dynamic consistency, obtains a subjective expected utility theory that the conditional preferences are connected through the Bayes rule. However, his representation only holds for the state-independent preferences.

%\todo[inline=yes, color=red]{(myself:A)
%lots of explanation, 

%B)and comparisons with Karni, Skiadas, Ghirartado, karni-schmeidler, Peter wakker, Anscombe auman, Jay lu.

%C)Comparison with my previous result, although it is just the rewriting of the result

%D)talk about the reason here we only consider the essential states(probabilities are positives
%}

\subsection*{Acknowledgments}
The authors gratefully acknowledge support
from Beyond Limits (Learning Optimal Models) through CAST (The Caltech Center for Autonomous Systems and Technologies) and partial support
from the Air Force Office of Scientific Research under awards number FA9550-18-1-0271 (Games for Computation and Learning) and FA9550-20-1-0358 (Machine Learning and Physics-Based Modeling and Simulation).

The first version of the paper was written during the first author's Ph.D. studies with many helpful comments from Federico Echenique and Kota Saito. The first author  thanks his Ph.D. advisors Jaksa Cvitanic, Federico Echenique, Kota Saito, and Robert Sherman. For helpful discussions, the first author thanks Itai Ashlagi, Kim Border, Martin Cripps, David Dillenberger, Drew Fudenberg, Simone Galperti, Michihiro Kandori, Igor Kopylov, Jay Lu, Fabio Maccheroni, Thomas Palfrey, Charles Plott, Luciano Pomatto, Antonio Rangel, Pablo Schenone, Omer Tamuz, and Leeat Yariv.

%%%%%%%%%%%%%%%%%%%%    References
%%%%%%%%%%%%%%%%%%%%   
%%%%%%%%%%%%%%%%%%%%   
\section*{References}
\printbibliography[heading=none]

%\bibliographystyle{plain}
%\bibliography{reference}

\newpage
\section{Appendix}

%%%%%%%%%%%%%%%%%%%%    Preliminaries

%\subsection{Preliminaries}
%This section contains basic definitions and results from convex analysis that are used in the proofs.

%\todo[inline, color=red]{Collinearity, Convexity, line, line segment between two points, affine hull, convex hull, affine combinations, and affinely independence should be defined!, convexity, relative interiors,relative boundary, segments [x,y], basic compactness results, proof of uniqueness under the affinity independence, Hyperplanes, and Separating Hyperplanes theorem , Farkas lemma(Kim Border, Infinite dimensional analysis) }

%%%%%%%%%%%%%%%%%%%%    Proof of Theorem 1

\subsection{Proof of Theorem \ref{main_theorem}}

The main part of the proof follows the steps of \cite{billot1}, with some twists.

The following two lemmas are the central ideas behind the proof. They help us to first define the function $w$ and then extend it from the binary sets to any finite-cardinality sets. 

\begin{lemma}\label{lem_main}
Select X as any nonempty set. Let $X^*$ denote the set of all nonempty finite subsets of X. Consider two functions $f_1,f_2: X^*\to \mathbb{R}^n$ that satisfy the strict weighted averaging axiom. Select four points $a,b,c,d$ in the space $X^*$ such that $a\cup b =c\cup d$ and $a\cap b=c\cap d=\emptyset$. If $\forall\  x\in \{a,b,c,d\} \  f_1(x)=f_2(x)$ and not all $\{f_1(a),f_1(b),f_1(c),f_1(d)\}$ are on a same line, then $f_1(a\cup b)=f_2(a\cup b)$.
\begin{proof}
Since $f_1$ satisfies the strict weighted averaging axiom and $a\cup b =c\cup d$ and $a \cap b = c \cap d = \emptyset$, thus $f_1(a\cup b)$ is on the line connecting $f_1(a),f_1(b)$. Also, since $a\cup b =c\cup d$ , $f_1(a\cup b)=f_1(c\cup d)$ should be on the line connecting $f_1(c)$ and $f_1(d)$. But $\{f_1(a),f_1(b),f_1(c),f_1(d)\}$ are not collinear, thus the line connecting $f_1(a)$ and $f_1(b)$ and the line connecting $f_1(c)$ and $f_1(d)$ can only intersect at most at a single point. But $f_1(a\cup b)$ is on the both lines, hence this point must be the unique intersection of them. 

Similarly, the same is true for $f_2$. This means $f_2(a\cup b)$ must be the unique intersection of the line passing through $f_2(a),f_2(b)$ and the line passing through $f_2(c),f_2(d)$. But since $\forall x\in \{a,b,c,d\} \  f_1(x)=f_2(x)$, $f_2(a\cup b)$ should be the unique intersection of the line passing through $f_1(a),f_1(b)$ and the line passing through $f_1(c),f_1(d)$. But we have already shown that $f_1(a\cup b)$ is also the unique intersection of the line passing through $f_1(a),f_1(b)$ and the line passing through $f_1(c),f_1(d)$. Thus, $f_1(a\cup b)=f_2(a\cup b)$.

\end{proof}
\end{lemma}
\begin{lemma}\label{lem_uniq}
Assume that $\{x,y,z\}$ are three points in X such that $f(x),f(y),f(z)$ are not collinear. Let $f$ satisfy the strict weighted averaging axiom and $f(\{x,y,z\})=a_1f(x)+a_2f(y)+a_3f(z)$, then $a_1/a_2$ must be independent of the choice of z, as long as $f(x),f(y),f(z)$ are not collinear.  Moreover, if $f(\{x,y\})=\lambda f(x)+(1-\lambda)f(y)$, then $\frac{a_1}{a_2}=\frac{\lambda}{1-\lambda}$.
\begin{proof}
Since $f(x),f(y),f(z)$ are not collinear, they should be affinely independent. Hence, $a_1,a_2,a_3$ are uniquely defined. 

By the strict weighted averaging axiom,  there exists $\lambda_1\in (0,1)$ such that $f(\{x,y,z\})=\lambda_1 f(\{x,y\})+(1-\lambda_1)f(\{z\})$. Again by the strict weighted averaging axiom there exists $\lambda\in (0,1)$ such that $f(\{x,y\})=\lambda f(x)+(1-\lambda)f(y)$. Hence, $f(\{x,y,z\})=\lambda_1 (\lambda f(x)+(1-\lambda)f(y))+(1-\lambda_1)f(\{z\})$. By affinely independence of $f(x),f(y),f(z)$, we should have $a_1=\lambda_1\lambda$ and $ a_2=\lambda_1(1-\lambda)$. This means that $\frac{a_1}{a_2}=\frac{\lambda}{1-\lambda}$, which means that $a_1/a_2$ is independent of the choice of z, as long as $f(x),f(y),f(z)$ are not collinear.
\end{proof}
\end{lemma}
\subsubsection{Proving the necessary and the uniqueness part}

Assume that the weight function $w$ exists. Therefore,  $f(A)=\frac{\sum_{x\in A}w(x)f(x)}{\sum_{x\in A} w(x)}$. It shows that if $A\cap B=\emptyset$, then $\  f(A\cup B)=\frac{\sum_{x\in A\cup B}w(x)f(x)}{\sum_{x\in A\cup B} w(x)}=\linebreak(\frac{\sum_{x\in A}w(x)}{\sum_{x\in A\cup B}w(x)})(\frac{\sum_{x\in A} w(x)f(x)}{\sum_{x\in A}w(x)})+(\frac{\sum_{x\in B}w(x)}{\sum_{x\in A\cup B}w(x)})(\frac{\sum_{x\in B} w(x)f(x)}{\sum_{x\in  B}w(x)})$. By defining $\lambda = \frac{\sum_{x\in A}w(x)}{\sum_{x\in A\cup B}w(x)}$, we have $f(A\cup B)= \lambda f(A)+(1-\lambda)f(B)$. Thus, the strict weighted averaging axiom satisfied. \\

Regarding the uniqueness of $w$, assume that there exist two $w_1,w_2$ such that $f(A)=\frac{\sum_{x\in A}w_1(x)f(x)}{\sum_{x\in A} w_1(x)}=\frac{\sum_{x\in A}w_2(x)f(x)}{\sum_{x\in A} w_2(x)}$. Since the range of $f$ is not a subset of a line, there exist at least three elements $x,y,z \in X$ such that $f(x),f(y),f(z)$ are not collinear. Thus, they are affinely independent.  Hence, $f(\{x,y,z\})=a_1f(x)+a_2f(y)+a_3f(z)$ has a unique solution $a_1,a_2,a_3$. Hence, there should be  an $\alpha$ such that $w_1(p)/w_2(p)=\alpha \ \forall p\in\{x,y,z\}$. We will show that for all other point $r\in X,\allowbreak \  w_1(r)/w_2(r)=\alpha$. 

Select a point $r\in X$, based on the assumption on $\{x,y,z\}$, there should be at least two points $u,v$ in $\{x,y,z\}$ such that $f(r),f(u),f(v)$ are not collinear. Without loss of  generality, assume that $\{u,v\}=\{x,y\}$. Since $f(r),f(x),f(y)$ are affinely independent, $f(\{x,y,r\})=b_1f(x)+b_2f(y)+b_3f(r)$ where $b_1,b_2,b_3$ are unique. Therefore, there exists $\beta$ such that $w_1(p)/w_2(p)=\beta \ \forall p\in\{x,y,r\}$. But notice that $\alpha=w_1(x)/w_2(x)=\beta$. Hence, we should have $w_1(r)/w_2(r)=\alpha$ and this is what we wanted to prove.
\\ 
\subsubsection{Proving the sufficiency part}
First, in order to define the function $w$, fix an element $x_{0}\in X$ and put $w(x_{0})=1$. Based on the strict weighted averaging axiom for any $y\in X \setminus \{x_{0}\}$ such that $f(y)\neq f(x_{0})$, we have a unique $\lambda\in (0,1)$ such that $f(\{x_{0},y\})=\lambda f(x_{0})+(1-\lambda)f(y)$. Let define $w(y)=\frac{1-\lambda}{\lambda}$.

To define the weight for any other $y\in X \setminus \{x_{0}\}$ with $f(y)= f(x_{0})$, we fix another point $z_{0}\in X \setminus \{x_{0}\}$ such that $f(z_{0})\neq f(x_{0})$. Since $f(x_{0})=f(y)$, we should have $f(y)\neq f(z_{0})$. By using the strict weighted averaging axiom, we know that there exists a unique $\lambda \in (0,1)$ such that $f(\{z_{0},y\})=\lambda f(z_{0})+(1-\lambda)f(y)$. Since the weight on $z_{0}$ has already been defined, we define the weight of $y$ such that $\frac{w(y)}{w(z_{0})}=\frac{1-\lambda}{\lambda}$. Thus, $w(y)=w(z_{0})\times\frac{1-\lambda}{\lambda}$.\\

In the rest of this section, we are going to prove that $w$ satisfies the representation of the theorem. It means that by defining $f^*(A)=\frac{\sum_{x\in A}w(x)f(x)}{\sum_{x\in A} w(x)}$,  we should have $f^*(A)=f(A)$.

%First, in Step 1 we prove that it works for 2 points, then in Step 2 we prove it for any choice of three points, and finally in Step 3, by induction on the cardinality of subsets of $X$, we show that  it works for any subset of $X$. \\ \\
%\textbf{Step 1:} In this step, we are going to prove that for any two points $\{r,s\}\in X,\  f(\{r,s\})=f^*(\{r,s\})$. We prove it in 5 separate cases.\\
%\textbf{Case 1:} If $f(r)=f(s)$ then by the strict weighted averaging axiom, $f(\{r,s\})=f(r)=f^*(r)=f^*(\{r,s\})$.\\ 
%\textbf{Case 2:} Assume that $r=x,\  f(s)\neq f(x)$. Based on the way we defined the weight over elements that are not $ f(x)$ , $f(\{r,s\})=f^*(\{r,s\})$.\\
%\textbf{Case 3:} Assume that $r,s $ are such that $f(r)\neq f(x)$ and $f(s) \neq f(x)$. If $f(r)=f(s)$, by Case 1 we are done, Otherwise consider $\{x,r,s\}$. We are going to first prove that $f\{x,r,s\}=f^*\{x,r,s\}$.

%\textbf{Case 4:}

%\textbf{Case 5:}

%%%%%%%%%%%%%%%%%%%%%%%%%%%%%%%%%%%%

First, in Step 1 we prove that the representation holds for any three points, as long as the three points under $f$ are not collinear. In Step 2, we prove that the representation holds for any two points. In Step 3, (which is not necessary, and we provide it for its simplicity to capture the main ideas of the main part) we prove that for three points the representation holds. Finally in Step 4, by using induction on the cardinality of subsets of $X$, we show that the representation holds for any subset of $X$.\\

\noindent\textbf{Step 1:} for any three points $r,s,t$ such that $f(r),f(s),f(t)$ are not collinear, we have $f(\{r,s,t\})=a_1f(r)+a_2f(s)+a_3f(t)$, where $a_i$ are unique. Note that, it is enough to prove that $\frac{a_1}{a_2}=\frac{w(r)}{w(s)}$, because in the same way, we can also get $\frac{a_2}{a_3}=\frac{w(s)}{w(t)},\ \frac{a_3}{a_1}=\frac{w(t)}{w(r)}$. There are two cases: 

\textbf{Case 1:} If $x_0,r,s$ are such $f(x_0),f(r),f(s)$ are not collinear then \linebreak$f(\{x_0,r,s\})=b_1f(x_0)+b_2f(r)+b_3f(s)$. Based on Lemma \ref{lem_uniq}, we know that $\frac{a_1}{a_2}=\frac{b_2}{b_3}$. But Again using the Lemma  \ref{lem_uniq} and the way we define $w$, we know that $\frac{b_1}{b_2}=\frac{1}{w(r)},\frac{b_1}{b_3}=\frac{1}{w(s)} $ which means that $\frac{b_2}{b_3}=\frac{w(r)}{w(s)} $. Hence, we have $\frac{a_1}{a_2}=\frac{w(r)}{w(s)}$. 

\textbf{Case 2:} If $x_0,r,s$ are such that $f(x_0),f(r),f(s)$ are collinear, in this case both $\{f(x_0),f(r),f(t)\}$ and $\{f(x_0),f(s),f(t)\}$ are not collinear.
By the same technique as the first case, we get that $\frac{a_1}{a_3}=\frac{w(r)}{w(t)}$ and $\frac{a_3}{a_2}=\frac{w(t)}{w(s)}$. Hence, it means that
$\frac{a_1}{a_2}=\frac{a_1}{a_3}\times\frac{a_3}{a_2}=\frac{w(r)}{w(t)}\times\frac{w(t)}{w(s)}=\frac{w(r)}{w(s)}$, which is what we wanted to prove.\\

\noindent\textbf{Step 2:} Assume that $r,s\in X$. We want to show that $f^*(\{r,s\})=f(\{r,s\})$. If $f(r)=f(s)$, then it is true. If $f(r)\neq f(s)$, then by the richness condition, there exists an element $t\in X$ such that $\{f(t),f(r),f(s)\}$ are not collinear. Based on Step 1, we know that $f(\{t,r,s\})=f^*(\{t,r,s\})$, also we have $f(t)=f^*(t),f(r)=f^*(r)$, and $f(s)=f^*(s)$. Notice that, based on the strict weighted averaging axiom, $f(\{r,s\})$ is on the line connecting $f(r)$ and $f(s)$. Also, it is on the line connecting $f(\{t,r,s\})$ and $f(t)$. The reason is that by the strict combination axiom, there exist a $\lambda \in (0,1)$ such that $f(\{t,r,s\})=\lambda f(t)+(1-\lambda)f(\{r,s\})$, which means that  $f(\{r,s\})$ is on the line connecting $f(\{t,r,s\})$ and $f(t)$. Similarly, everything holds for $f^*$ which means that $f^*(\{r,s\})$ is on the line connecting $f^*(r)$ and $f^*(s)$ and also it is on the line connecting $f^*(\{t,r,s\})$ and $f^*(t)$. Since $\{f(t),f(r),f(s)\}$ are not collinear the intersection of two line can have at most one solution and since $f(\{t,r,s\})=f^*(\{t,r,s\})$, $f(t)=f^*(t)$, $f(r)=f^*(r)$, and $f(s)=f^*(s)$ then by a similar argument as Lemma \ref{lem_main}, we should have a unique intersection, which satisfies $f^*(\{r,s\})=f(\{r,s\})$. This is what we wanted to prove.\\ \\
\textbf{Step 3:} (This part is the tricky part, we provide it to capture the main ideas. We will use the same technique in Step 4) We are going to prove that for all three point $r,s,t$ we have $f^*(\{r,s,t\})=f(\{r,s,t\})$.There are two separate cases to be considered.

\textbf{Case 1:} If $f(r),f(s),f(t)$ are not collinear, then by Step 1, it is correct.

\textbf{Case 2:} Assume that $f(r),f(s),f(t)$ are collinear. If all of them are the same, then by strict weighted averaging axiom $f^*(\{r,s,t\})=f(\{r,s,t\})$. Hence, assume that they are not all the same. 

Without loss of  the generality, assume that $f(s)\neq f(r), f(s)\neq f(t)$. Based on the richness condition of $f$, we should have a point $v\in X$ such that not all $f(v),f(r),f(s)$, and $f(t)$ are collinear. Note that, $f(v),f(r),f(s)$ are not collinear. Similarly, $f(v),f(s),f(t)$ are not collinear. Based on Case 1, we know that $f^*(\{v,r,s\})=f(\{v,r,s\})$ and $f^*(\{v,s,t\})=f(\{v,s,t\})$. Also, we know that $f^*(v)=f(v)$, $f^*(r)=f(r)$, $f^*(s)=f(s)$, and $f^*(t)=f(t)$. Using the strict weighted averaging axiom, we know that $f(\{v,r,s,t\})$ is on the intersection of the line passing through $f(\{v,r,s\})$ and $f(t)$, and the line passing through the $f(\{v,s,t\})$ and $f(r)$. Also, note that not all of $f(\{v,r,s\})$, $f(\{v,s,t\})$, $f(r)$, and $f(t)$ are collinear, since otherwise $f(v)$ must be on the line connecting $f(r)$ and $f(s)$. Similarly, we have the same properties for $f^*(\{v,r,s,t\})$. Based on the argument of the Lemma \ref{lem_main}, we have $f(\{v,r,s,t\})=f^*(\{v,r,s,t\})$.

By using the strict weighted averaging axiom, we know that $f(\{r,s,t\})$ is on the line passing through $f(\{v,r,s,t\})$ and $f(v)$, since there exists $\lambda \in (0,1)$ such that $f(\{v,r,s,t\})=\lambda f(\{r,s,t\})+(1-\lambda)f(v)$. Again, by using the strict weighted averaging axiom, we know that $f(\{r,s,t\})$ is on the line passing through $f(\{r,s\})$ and $f(t)$. Also, the same holds for $f^*$. Moreover, we have $f(\{v,r,s,t\})=f^*(\{v,r,s,t\})$, $f(\{r,s\})=f^*(\{r,s\})$, $f(v)=f^*(v)$, and $f(t)=f^*(t)$. Also, not all $f(\{v,r,s,t\})$, $f(\{r,s\})$, $f(v)$, and $f(t)$ are on a same line, since otherwise $f(v)$, $f(r)$, $f(s)$, and $f(t)$ are collinear which is not correct. As a result, based on the argument of lemma \ref{lem_main}, we have $f^*(\{r,s,t\})=f(\{r,s,t\})$. The latter is what we wanted to prove.\\ 
\\
\textbf{Step 4 (The main Step):} Up to here, we prove that for any $A\in X^*$ if $|A|\leq 3$ then $f^*(A)=f(A)$. To complete the proof, we will use an induction on the cardinality of $A$. Assume that for all $A\in X^*$ with $ |A|\leq k$ we have $f^*(A)=f(A)$. We are going to show that for all $A\in X^*$ with $ |A|=k+1$, we have  $f^*(A)=f(A)$.

Fix a subset $A $ with $|A|=k+1$. Assume that $A=\{x_1,\ldots,x_{k+1}\}$. There are two separate cases to be considered.

\textbf{Case 1:} Assume that not all $\{f(x_i)\}_{i=1}^{k+1}$ are collinear. Note that, by the induction hypothesis, $\forall x\in A$ and $\forall \  B\in 2^{A\setminus\{x\}}$  we have $f(B)=f^*(B)$. Define $line(f(x),f(A\setminus\{x\}))$ as the line passing through $f(x)$ and $f(A\setminus\{x\})$ for the case where $f(x)\not= f(A\setminus\{x\})$. However, if $f(x)=f(A\setminus\{x\})$, then define it as the single point $f(x)$.

If there exists $ x\in A$ such that $f(x)=f(A\setminus\{x\})$, then based on the strict weighted averaging axiom , there exists $\lambda \in (0,1)$ such that $f(A)=\lambda f(x)+(1-\lambda)f(A\setminus \{x\})=f(x)$. Similarly, $f^*(A)=f^*(x)$. But, we know that $f(x)=f^*(x)$, which means that $f(A)=f^*(A)$. 

If $\forall \ x\in A \ f(x)\not=f(A\setminus\{x\})$, then there exist $x,y\in A$ such that $f(x), f(A\setminus \{x\})$, $f(y), f(A\setminus \{y\})$ are not collinear. Because, otherwise all $f(x_i)$ are on the $f(x),f(A\setminus \{x\})$, which cannot be correct since we assumed that not all $\{f(x_i)\}_{i=1}^{k+1}$ are collinear. Considering $x,y\in A$ such that $f(x),f(A\setminus \{x\}),f(y),f((A\setminus \{y\})$ are not collinear, based on the strict weighted averaging axiom we know that $f(A)$ is on $line(f(x),f(A\setminus \{x\}))$. Also, it must be on $line(f(y),f(A\setminus \{y\}))$. Similarly, by the strict weighted averaging axiom when applied to $f^*$, we know that $f^*(A)$ is on $line(f^*(x),f^*(A\setminus \{x\}))$. Moreover, it must be on the $line(f^*(y),f^*(A\setminus \{y\}))$. 

Since (1) $f(x)=f^*(x)$,  $f(A\setminus \{x\})=f^*(A\setminus \{x\})$, $P(y)=P^*(y)$, $P(A\setminus \{y\})=P^*(A\setminus \{y\})$ and, (2) not all $f(x),f(A\setminus \{x\}),f(y)$, and $f((A\setminus \{y\})$ are collinear, based on the Lemma \ref{lem_main}, we have $f^*(A)=f(A)$.
Hence in the case that not all $\{f(x_i)\}_{i=1}^{k+1}$ are collinear, we showed that $f^*(A)=f(A)$.\\

\textbf{Case 2:} Assume that $\{f(x_i)\}_{i=1}^{k+1}$ are collinear. Without loss of generality, assume that $f(x_1),f(x_{k+1})$ are the two extreme points on the line that contains them, which means that all other points are between this two. 

If $f(x_1)=f(x_{k+1})$,then all $\{f(x_i)\}_{i=1}^{k+1}$ are the same. Using the strict weighted averaging axiom, it shows that $f(A)=f(x_1)=f^*(x_1)=f^*(A)$.

If $f(x_1)\neq f(x_{k+1})$, based on the richness condition of the aggregation rule $f$, we can select a point $y\in X\setminus A$ such that not all $f(y),f(x_1)$, and $f(x_{k+1})$ are collinear. Based on the previous Case 1, we know that $f(y,x_1,\ldots,x_k)=f^*(y,x_1,\ldots,x_k) $, since we have proved that $f$ and $f^*$ are coincided for any $k+1$ not collinear points. Similarly, we have $f(y,x_2,\allowbreak \ldots,x_{k+1})=f^*(y,x_2,\ldots,x_{k+1}) $. 

Using the strict weighted averaging axiom, $\allowbreak f(\{y,x_1,\ldots,x_{k+1}\})$ is on the \linebreak$\allowbreak line(\allowbreak f(\{y,x_1,\ldots,x_k\})\allowbreak,f(x_{k+1}))$. It is also on the $\allowbreak line(f(\{y,x_2,\allowbreak \ldots,x_{k+1}\}),f(x_1))$. Also, not all $f(\{y,x_1,\ldots,x_k\})$, $f(x_{k+1}))$, $\allowbreak f(\{y,x_2,\allowbreak \ldots,\allowbreak x_{k+1}\allowbreak \})$, $f(x_1))$ are collinear, since $f(\{y,x_1,\allowbreak \ldots,x_k\})$ cannot be on $\allowbreak line(f(x_1),\allowbreak f(x_{k+1}))$ otherwise $f(y)$ must be on that line which is not correct. Similarly, everything holds for the $f^*$.

Since, $f(\{y,x_1,\allowbreak\ldots,x_k\})= f^*(\{y,x_1,\ldots,x_k\})$, $f(x_{k+1})=f^*(x_{k+1})$,  $f(\{y,x_2,\ldots,x_{k+1}\}) = f^*(\{y,x_2,\ldots,x_{k+1}\})$, and $f(x_1)=f^*(x_1)$ then again by using Lemma \ref{lem_main}, we get $f(\{y,x_1,\ldots,x_{k+1}\})=f^*(\{y,x_1,\ldots,x_{k+1}\})$. 

The point $f(\{x_1,\ldots,x_{k+1}\})$ is on the $line(f(x_1),f(x_k))$. It is also on the $line(f(y)\allowbreak,f(\{y,x_1,\ldots,x_{k+1}\}))$, since by the strict weighted averaging axiom $f(\{y,x_1,\ldots,x_{k+1}\})=\lambda f(y)+(1-\lambda)f(\{x_1,\ldots,x_{k+1}\})$ for some $\lambda \in (0,1)$. Similarly, the same holds for $f^*$. Finally, since (1) $f(y)=f^*(y)$, $f(\{y,x_1,\allowbreak\ldots,\allowbreak x_{k+1}\allowbreak \})=\allowbreak f^*(\{y,x_1,\ldots,x_{k+1}\})$, $f(x_1)=f^*(x_1)$ and, (2) $f(x_{k+1})=f^*(x_{k+1})$ and $f(x_1),f(x_{k+1})$, and $f(y)$ are not collinear, then using the same types of arguments in Lemma \ref{lem_main}, we get $f(A)=f^*(A)$. This is what we wanted to prove. \\

Hence, for all $A\in X^*$ with cardinality $k+1$, we have $f(A)=f^*(A)$. Based on the induction, we have $f(A)=f^*(A)$ for all $A\in X^*$. This completes the proof.\\
\qed

%%%%%%%%%%%%%%%%%%%%    Proof of Theorem 2

\subsection{Proof of Theorem \ref{general_thm}}
There are couple of steps in the proof.
\noindent\textbf{Defining the weak order:}

\noindent\textbf{Step 1:} First, we define a binary relation $\succcurlyeq$ over every two different elements $x,y\in X$ by:

\textbf{Case 1:} If $f(x)\neq f(y)$, we define $x\succcurlyeq y\iff f(\{x,y\})\neq f(y)$.

\textbf{Case 2:} If $f(x)= f(y)$, then by the strong richness condition, we select another point $z\in X$, such that $f(\{x,z\})\notin \{f(x),f(z)\}$. Hence, we have $f(z)\neq f(y)$. In this case, we define $x\succcurlyeq y\iff f(\{z,y\})\neq f(y)$.

To obtain reflexivity, for any $x\in X$, we define $x\succcurlyeq x$.

\noindent\textbf{Step 2:}
We prove that $\succcurlyeq$ is a weak order. The reflexivity and the completeness are trivial. We only need to establish the transitivity. Assume that $x\succcurlyeq y, y\succcurlyeq z$. We will show that $x\succcurlyeq z$. 

The proof is by contradiction. Therefore, assume that $z\succ x$.

\textbf{Case 1:} Assume that $f(x),f(y),f(z)$ are non-collinear.  Since $z\succ x$, based on the way we defined $\succcurlyeq$, we have $f(x,z)=f(z)$. 

Consider the coalition $\{x,y,z\}$. By using the weighted averaging axiom over the sub-coalitions $\{x,z\}$ and $\{y\}$, the vector $f(\{x,y,z\})$ should be on the line joining $f(y)$ and $ f(\{x,z\})$ (which is the same as $f(z)$). Similarly, by considering the sub-coalitions $\{x,y\}$ and $\{z\}$,  $f(\{x,y,z\})$ should be on the line passing through $f(\{x,y\})$ and $f(z)$. Since $f(\{x,y\})\neq f(y)$ and $f(\{x,y\}),f(y),f(z)$ are non-collinear, we have $f(\{x,y,z\})=f(z)$. However, by considering the sub-coalitions $\{y,z\}$, $\{x\}$ and the fact that $f(\{y,z\})\neq f(z)$, this cannot be happen. Therefore,  $x\succcurlyeq z$.

\textbf{Case 2:} Assume that $f(x),f(y),f(z)$ are collinear. By using the strong richness condition, we can select a point $u\in X$ such that $f(u)$ is not on the line passing through $f(x),f(y),f(z)$, and also $f(\{u,x\})\notin \{f(x),f(u)\}$ (this means that $x\sim u$). First, using Case 1, by considering the coalitions $\{u,x,y\}$, $\{u,x,z\}$, we have $u\succcurlyeq y$ and $z\succ u$. Since $\{u,y,z\}$ are non-collinear, by using case 1, we have $z\succ u, u\succcurlyeq y\Rightarrow z\succ y$. But this is a contradiction. Therefore, $x\succcurlyeq z$.\\

\noindent\textbf{The main part: proving  $f(A)=f(M(A,\succcurlyeq))$}.

Up to here, we show that $\succcurlyeq$ is a weak order. Next, we will show that for any coalition $A\in X^*$ we have $\ f(A)=f(M(A,\succcurlyeq))$.

We use the letter $H$ for the highest-ordered elements of $A$, and $L$ for the rest. In other words, $H:=M(A,\succcurlyeq), \ L:=A\setminus H$. The proof is by a double-induction on the cardinality of $H$ and $L$. In Step 1, we will show that if $x\in X$ and $L\in X^*$ are such that $\forall z\in L: \ x\succ y$, then we should have $f(\{x\}\cup L)=f(x)$. 

In Step 2, we show that for a given coalition $H\in X^*$, where all elements of $H$ are in the same equivalence class, and for all $L\in X^*$, if for all $x\in H, \ y\in L: x\succ y $, then we have $f(H\cup L)=f(H)$. Using these two steps, we will finish the proof.\\
 
\noindent\textbf{Step 1:} Fix an element $x\in X$. By induction on the cardinality of $L$, where $\forall y\in L, \ x\succ y$, we prove that $f(\{x\}\cup L)=f(x)$. 
 
We have already proved the case where $|L|=1$. Assume that for all $|L|\leq k$ the result is correct. We will show that for all $L$ with $|L|=k+1$, the result is also correct.
 
Fix a coalition $L$ with $|A|=k+1$ and such that $\forall y\in L, \ x\succ y$. Assume that $A=\{y_1,\ldots,y_{k+1}\}$.

If for all $y\in L: f(y)=f(x)$, then using the weighted averaging axiom $f(\{x\}\cup L)=f(x)$, which is what we wanted to prove. Similarly, if $f(L)= f(x)$, we have $f(\{x\}\cup L)=f(x)$.

Therefore, consider the case that not all of them are the same and $f(L)\neq f(x)$.

Using our definition of the $line$ in the proof of Theorem \ref{main_theorem}, for each $y\in L$ we consider $line(f(\{x\}\cup (L\setminus \{y\})), f(y))$. Using the weighted averaging axiom, for all $y\in L,\  f(\{x\}\cup L)\in line(f(\{x\}\cup L\setminus \{y\}), f(y)) $. By using the induction hypothesis $f(\{x\}\cup L\setminus \{y\})=f(x)$. Therefore, $\forall y\in L,\  f(\{x\}\cup L)\in line(f(x), f(y))$. 

Similar to the proof of Theorem \ref{main_theorem}, we consider two separate cases.\\

\textbf{Case 1:} Consider the case where there exist two elements $y_1,y_2\in L$ such that  $f(x),f(y_1),f(y_2)$ are non-collinear. We use the same technique as in the proof of Theorem \ref{main_theorem}. 

We know that $f(\{x\}\cup L)\in line(f(x), f(y_1)) $ as well as $f(\{x\}\cup L)\in line(f(x), f(y_2))$. Moreover, we know that $f(x),f(y_1),f(y_2)$ are non-collinear. Therefore, the intersection of the two lines, should be $f(x)$. This shows that
$f(\{x\}\cup L)=f(x)$.

\textbf{Case 2:} If all the vectors $f(x), f(y_1),\ldots ,f(y_{k+1})$ are collinear. In this case, the idea is to add a point $x'\in X$ such that $x'\sim x$ and $f(x')$ is not on the line containing all $f(x), f(y_1),\ldots ,f(y_{k+1})$. This is possible because of the strong richness condition. By using the transitivity of the  $\succcurlyeq $, $\forall y\in L:\   x'\succ y$. 

Fix a point $y_0\in L$ such that $f(y_0)\neq f(x)$. This is possible since we already assumed that not all $f(y)$ ,with $y\in L$, are the same as $f(x)$.

Consider the coalition $\{x\}\cup\{x'\}\cup L$. By using the weighted averaging axiom and the sub-coalitions $\{x,L\setminus \{y_0\}\}, \ \{x',y_0\}$, we have  $f(\{x\}\cup\{x'\}\cup L)\in line(f(\{x,L\setminus \{y_0\}\}), f(\{x',y_0\}))$. Using the induction hypothesis, $f(\{x,L\setminus \{y_0\}\})=f(x)$ and $f(\{x',y_0\})=f(x')$. Therefore, $f(\{x\}\cup\{x'\}\cup L)\in line(f(x),f(x'))$. 

Next, we show that $f(\{x\}\cup\{x'\}\cup L)\neq f(x)$. Since $x\sim x'$ and $f(x)\neq f(x')$, we have $f(\{x,x'\})\neq f(x')$. Moreover, based on the way we selected the point $x'$, $f(\{x,x'\})$ is not on the line containing $f(x)$ and  $\{f(y) | y\in L\}$. Consider the partition of $\{x\}\cup\{x'\}\cup L$ into $\{x,x'\}$ and $L$. Based on the choice of the $L$, at the beginning of Step 1,  $f(L)\neq f(x)$. Since  $f(\{x,x'\})\neq f(x)$, $f(L)\neq f(x)$, and $f(x),f(x'),f(L)$ are non-collinear, we have $f(\{x\}\cup\{x'\}\cup L)\neq f(x)$.

Finally, by partitioning $\{x\}\cup\{x'\}\cup L$ into $\{x'\}$ and $\{x\}\cup L $, the weighted averaging axiom results in $f(\{x\}\cup\{x'\}\cup L) \in line(f(\{x\}\cup L), f(x'))$. Therefore, $f(\{x\}\cup L)$ is on the line joining $f(\{x\}\cup\{x'\}\cup L)$ and $f(x')$. However, we have already shown that $f(\{x\}\cup\{x'\}\cup L)$ is on the line passing through $f(x)$ and $f(x')$. Thus, $f(\{x\}\cup L)$ should be on the line joining $f(x)$ and $f(x')$. 
However, the only intersection of $line(f(x),f(x'))$ and the line containing all the points $f(x),f(y_1),\dots, f(y_{k+1})$ is the point $f(x)$. Thus, $f(\{x\}\cup L)=f(x)$, which completes the proof.\\

\noindent\textbf{Step 2:} In this step, by using induction on the cardinality of the set $H$, in which all elements have the same order, we show that for any coalition $L$ if all elements of the set $L$ have lower orderings compared to the elements of $H$, then  we should have $f(H\cup L)=f(H)$.

Fix a set $L$. Based on Step 1, we know that for any $x\in X$ such that $\forall y\in L: x\succ y$,  we should have $f(\{x\}\cup L)=f(x)$. This is the starting point of our induction. Assume that for all ${|H|}=k$, we have $f(H\cup L)=f(H)$. We will show that for any $|H|=k+1$, we have $f(H\cup L)=f(H)$.\\

For any $x\in H$, by the weighted averaging axiom over the sub-coalitions $\{x\}\cup L$ and $H\setminus\{x\}$, we have $f(H\cup L)\in line(f(\{x\}\cup L), f(H\setminus \{x\}))$. Based on step 1, we know that $f(\{x\}\cup L)=f(x)$. Therefore, $f(H\cup L)\in line(f(x), f(H\setminus \{x\}))$. Similarly, by the weighted averaging axiom over the coalition $H$ and its sub-coalitions $\{x\},H\setminus\{x\}$, we should have $f(H)\in line(f(x), f(H\setminus \{x\}))$.
Consider two cases:\\

\textbf{Case 1:} Consider the case in which not all members of $\{f(x)|x\in H\}$ are collinear. Hence, there should be at least two elements $x,y\in H$ that $f(x),f(H\setminus\{x\}),f(y),$ and $f(H\setminus\{y\})$ are not collinear. Therefor, the intersection of the lines joining $f(x),f(H\setminus\{x\})$ and the line joining $f(y),f(H\setminus\{y\})$ can have at most one intersection. Since $f(H)$ is on both lines, the unique intersection should be $f(H)$. But $f(H\cup L)$ is also on both lines. Hence, we should have $f(H\cup L)=f(H)$, which completes this case.\\

\textbf{Case 2:} Consider the case where all members of the set $\{f(x)|x\in H\}$ are on a line. By using the strong richness condition, there exists an element $x'\in X$ such that $f(x')$ is not on that line.
We consider the coalition $\{x'\}\cup H\cup L$. By using the weighted averaging axiom over the sub-coalitions $\{x'\}\cup L$ and $H$, we should have $f(\{x'\}\cup H\cup L)$ on the line joining $f(\{x'\}\cup L)$ and $f(H)$. By the induction hypothesis, $f(\{x'\}\cup L)=f(x')$. Hence, we should have $f(\{x'\}\cup H\cup L)\in line(f(x'), f(H))$. 

Similarly, by partitioning the set $\{x'\}\cup H\cup L$ into $H\cup L$ and $\{x'\}$, we have $f(\{x'\}\cup H\cup L)\in line(f(x'), f(H\cup L))$. 

Select an element  $x_1\in H$. By partitioning the coalition $\{x'\}\cup H\cup L$ between the sub-coalitions $\{x_1,x'\}$ and $(H\setminus\{x_1\}) \cup\ L$, using the weighted averaging axiom, we obtain $f(\{x'\}\cup H\cup L)\neq f(\{x'\})$.

Finally, using (1) $f(\{x'\}\cup H\cup L)\in line(f(x'), f(H))$, and (2) $f(\{x'\}\cup H\cup L)\in line(f(x'), f(H\cup L))$, and (3) $ f(\{x'\}\cup H\cup L)\neq f(x')$, we have \linebreak $line(f(x'), f(H))=\allowbreak line(f(x'), f(H\cup L))$. But the intersection of the last line with the line containing all the elements of $H$, can have at most one intersection. Therefore, $f(H\cup L)=f(H)$,
which completes the proof.\\

\noindent\textbf{Completing the proof:}

Consider a coalition $H$ where all elements have the similar order. We consider any two disjoint sub-coalitions $H_1,H_2\in H$, where $f(H_1)\neq f(H_2)$. Using the same technique of the previous part, we have $f(H_1\cup H_2)\neq f(H_1)$.

By using the result of Theorem \ref{main_theorem}, we can get the appropriate representation in each equivalence class. Also by using the result of the previous part, $f(A)=f(M(A,\succcurlyeq))$. The combination of these two results completes the proof.
\qed

%%%%%%%%%%%%%%%%%%%%    Proof of theorem 3

\subsection{Proof of Theorem \ref{proposition_contin}}
The following two lemmas help us proving the theorem.
\begin{lemma}
Given any two linearly independent vectors $v_1,v_2$ in $\mathbb{R}^n$, there exists a neighborhood of $v_1$ that any vector in that neighborhood is linearly independent of $v_2$. More generally, given any $m$ vectors $\{v_1,\ldots, v_m\}$ such that $v_1$ is not in the linear space generated by the rest of the points, then there exists a neighbor of $v_1$ such that any point in that neighborhood is not in $\text{span}(\{v_2,\ldots,v_m\})$.
\end{lemma}
\begin{proof}
Since $K=\text{span}(\{v_2,\ldots,v_m\})$ is a closed set that is disjoint from the vector $v_1$, the distance between $v_1$ and $K$ should be nonzero. Hence, there exists a neighborhood of $v_1$ (for example the ball with radios $dist(v_1,K)/2$ around $v_1$) disjoint from $K$. As a result, any point in that neighborhood is not in $\text{span}(\{v_2,\ldots,v_n\})$.
\end{proof}

\begin{lemma}\label{lin_indep}
Let $v_1,v_2 \in \mathbb{R}^n$ be two linearly independent vectors and $v=\alpha v_1 +(1-\alpha) v_2$, for some $\alpha \in [0,1]$, is a vector between $v_1,v_2$. If the vectors $v_n=\alpha_n v_1+(1-\alpha_n) v_2^n$ are such that $\alpha_n\in[0,1]$, $v_2^n\rightarrow v_2$, and $v_n\rightarrow v$, then $\alpha_n\to \alpha$.
\end{lemma}
\begin{proof}
We prove it by contradiction. If it is not the case, there exists a subsequence $\alpha_{n_k}$ of $\alpha_n$ and some $\epsilon >0$, such that $\forall \  n_k: \ \alpha_{n_k}\notin B_\epsilon(\alpha)$. Based on compactness of $[0,1]$, there exist a subsequence $\alpha_{n_{k_j}}$ of $\alpha_{n_k}$ that is convergent to some $\beta\in [0,1]$. Since $\alpha_{n_k}\notin B_\epsilon(\alpha)$, we have $\beta\not= \alpha$. Based on the assumption of the lemma, since the sequence $v_n$ is convergent to $v$, the subsequence $v_{n_{k_j}}$ also converges to $v$. Similarly, $v_2^{n_{k_j}}$ converges to $v_2$. Hence, $v_{n_{k_j}}=\alpha_{n_{k_j}} v_1+(1-\alpha_{n_{k_j}})v_2^{n_{k_j}}\rightarrow \beta v_1+(1-\beta)v_2$ and $v_{n_{k_j}}\rightarrow v$. As a result, $\beta v_1+(1-\beta)v_2=v=\alpha v_1 +(1-\alpha) v_2$.  However, since $v_1,v_2$ are linearly independent, $\alpha$ and $\beta$ should be the same, which is a contradiction. The contradiction shows that $\alpha_n\to \alpha$.
\end{proof}

Using the lemmas mentioned above, we will complete the proof. Based on Theorem \ref{main_equation2}, there exist a unique weak order $\succcurlyeq$ and a weight function $ w: X\to \mathbb{R}_{++}$ such that for any $A\in X^*$\[f(A)=\frac{\sum\limits_{x\in M(A,\succcurlyeq)}w(x)f(x)}{\sum\limits_{x\in M(A,\succcurlyeq)} w(x)}.\]\\

Let $x\in X$ be any given point. We need to prove that the weight function is continuous around $x$ and any point close enough to $x$ has the same order, respect to the weak order $\succcurlyeq$, as $x$.

To complete the proof, assume that $x_n\in X$ and $x_n\to  x$. We are going to prove that:\\ 1) $w(x_n)\to w(x)$, \\2) $\exists N \in \mathbb{N}$ such that for all $n>N: $ $x_n\sim x$.\\Proving these two completes the proof.

Based on the strong richness condition, there should be a point $y\in X$ such that (1) $f(x), f(y)$ are linearly independent, and (2) $f(\{x,y\})=\frac{w(x)f(x)+w(y)f(y)}{w(x)+w(y)}$, which means that $x\sim y $. The reason is that by the strong richness condition, there should be at least two other points $y,z$ with the same order as $x$, such that not all of $f(x),f(y),$ and $f(z)$ are collinear. This means that $f(x)$ and at least one of $f(y)$ or $f(z)$ should be linearly independent. Without loss of generality, we assume that $f(x)$ and $f(y)$ are linearly independent.\\

Given any two points $a,b\in X$, we define the function $\boldsymbol{1}_a(b)$ as follows: 
  \begin{equation*}
   \boldsymbol{1}_a(b) =
    \begin{cases}
      1 & \text{if } b \succcurlyeq a,\\
      0 & \text{Otherwise} .
    \end{cases}
\end{equation*}

Consider the sequence of vectors $f(\{x_n,y\})$. By Theorem \ref{general_thm}, we have $\allowbreak f(\{x_n,y\})\allowbreak =\allowbreak\frac{\boldsymbol{1}_y(x_n)w(x_n)f(x_n)+\boldsymbol{1}_{x_n}(y)w(y)f(y)}{\boldsymbol{1}_y(x_n)w(x_n)+\boldsymbol{1}_{x_n}(y)w(y)}\allowbreak=\frac{\boldsymbol{1}_y(x_n)w(x_n)}{\boldsymbol{1}_y(x_n)w(x_n)+\boldsymbol{1}_{x_n}(y)w(y)}f(x_n)+\frac{\boldsymbol{1}_{x_n}(y)w(y)}{\boldsymbol{1}_{x_n}(y)w(y)+\boldsymbol{1}_y(x_n)w(x_n)}f(y)$. Based on continuity of the aggregation rule $f$, $f(x_n)\to f(x)$ and $f(\{x_n,y\})\to f(\{x,y\})$. Since $f(x)$ and $f(y)$ are linearly independent, all conditions of  Lemma \ref{lin_indep} are satisfied. Hence, we have  $\frac{\boldsymbol{1}_{x_n}(y)w(y)}{\boldsymbol{1}_y(x_n)w(x_n)+\boldsymbol{1}_{x_n}(y)w(y)}\to \frac{w(y)}{w(x)+w(y)}$
and $\frac{\boldsymbol{1}_y(x_n)w(x_n)}{\boldsymbol{1}_y(x_n)w(x_n)+\boldsymbol{1}_{x_n}(y)w(y)}\to \frac{w(x)}{w(x)+w(y)}$. Since both $w(x)$ and $w(y)$ are strictly positive, we should have $\boldsymbol{1}_y(x_n)\to 1$ and similarly $\boldsymbol{1}_{x_n}(y)\to 1$. This means that for large $n$, $x_n\sim y$. Since $y\sim x$, for large $n$ we have $x_n\sim x$. This complete part 2 of the proof.

For the part 1, since we have already proved that for large $n$, $x_n\sim x\sim y$, the convergence  $\frac{\boldsymbol{1}_{x_n}(y)w(y)}{\boldsymbol{1}_y(x_n)w(x_n)+\boldsymbol{1}_{x_n}(y)w(y)}\to \frac{w(y)}{w(x)+w(y)}$ becomes $\frac{w(y)}{w(x_n)+w(y)}\to \frac{w(y)}{w(x)+w(y)}$. This means that  $w(x_n)\to w(x)$, which proves that $w$ is continuous at $x$. 

Proving part 1 and 2 complete the proof.\\
\qed

%%%%%%%%%%%%%%%%%%%%    Proof of Theorem 4

%%%%%%%%%%%%%%%%%%%%    Proof of Theorem 5

\subsection{Proof of Proposition \ref{time_stationary}}
There are a couple of steps to prove the result.

\noindent\textbf{Step 1:}
Assume that all signals arrive at time $1$. By using Corollary \ref{CBFT_corol}, there exists a unique (up to multiplication) weight function $w:X^*\to \mathbb{R}_{++}$, such that for all $A\in X^*$, $f(A,1)=\frac{\sum\limits_{x\in A}w(x)f(x)}{\sum\limits_{x\in A} w(x)}$. By using the uniqueness of $w$ and the stationarity axiom, for any constant time shift $c$ and for all $A\in X^*$ we have:   $$f(A,c)=\frac{\sum\limits_{x\in A}w(x)f(x)}{\sum\limits_{x\in A} w(x)}.$$ 

Consider two signals $x_0,y_0\in X$, where $f(x_0)\neq f(y_0)$. Let the timing of $x_0,y_0$ be $T_{\{x_0, y_0\}}(x_0)=1, T_{\{x_0, y_0\}}(y_0)=2$. Using the strict weighted averaging axiom, there exists a  $\lambda\in (0,1)$  where $f(\{x_0,y_0\},T_{\{x_0, y_0\}})=\lambda f(x_0)+(1-\lambda)f(y_0)$. We define $q$ such that $\frac{1-\lambda}{\lambda}=q\times\frac{w(y_0)}{w(x_0)}$.  

In the rest of the proof, we show that these choices of $w,q$ attain the representation of Proposition \ref{time_stationary}.

\noindent\textbf{Step 2:} We show that for any signal $z\in X$, the representation holds for the coalition $\{x_0,z\}$ and for the timing function $T_{\{x_0, z\}}(x_0)=1, T_{\{x_0, z\}}(z)=2$.

\textbf{Case 1:} Consider any signal $z\in X$, such that $\{f(x_0),f(y_0),f(z)\}$ are not collinear. We form the coalition $\{x_0,y_0,z\}$ with the timing $T_{\{x_0, y_0, z\}}(x_0)=1, T_{\{x_0, y_0, z\}}(y_0)=2$, $T_{\{x_0, y_0, z\}}(z)=2$. Using the strict weighted averaging axiom, by considering the sub-coalitions $\{x_0\}$ and $\{y_0,z\}$  and the fact that $y_0$ and $z$ has the same timing, Lemma \ref{lem_uniq}, in the proof of Theorem \ref{main_theorem}, shows that the representation holds for the coalition $\{x_0,z\}$ with the timing $T_{\{x_0, z\}}(x_0)=1, T_{\{x_0, z\}}(z)=2$.

\textbf{Case 2:} Consider any signal $z\in X$, such that $\{f(x_0),f(y_0),f(z)\}$ are collinear. By the richness condition, there exists a signal $z'\in X$ such that $\{f(x_0),f(y_0),f(z),f(z')\}$ are not collinear. We consider the timing $T_{\{x_0, y_0, z, z'\}}(x_0)=1, T_{\{x_0, y_0, z, z'\}}(y_0)=2$, $ T_{\{x_0, y_0, z, z'\}}(z)=2, T_{\{x_0, y_0, z, z'\}}(z')=2$.
The representation holds for the sub-coalitions $\{x,y,z'\}$ (by Case 1) and $\{y,z,z'\}$ (since all have the same timing). Thus, by applying Lemma \ref{lem_uniq} first on $\{y,z,z'\}$ and then on $\{x,y,z'\}$, we can show that the representation holds for the coalition $\{x_0,z\}$ with the timing $T_{\{x_0, z\}}(x_0)=1, T_{\{x_0, z\}}(z)=2$.

\noindent\textbf{Step 3:} We show that the representation holds for any two signals $u,v\in X$ with the timing function $T_{\{u, v\}}(u)=1, T_{\{u, v\}}(v)=2$.

\textbf{Case 1:} If  $\{f(x_0),f(u),f(v)\}$ are non-collinear, then we consider the timing function $T_{\{x_0, u, v\}}(x_0)=1, T_{\{x_0, u, v\}}(u)=1 , T_{\{x_0, u, v\}}(v)=2$. By applying Lemma \ref{lem_uniq} twice on  $\{x_0, u\}$ and $\{x_0, v\}$ with their corresponding timings, we can show that the representation should holds for $u,v\in X$ with the timing function $T_{\{u, v\}}(u)=1, T_{\{u, v\}}(v)=2$.

\textbf{Case 1:} If  $\{f(x_0),f(u),f(v)\}$ are collinear, then by the richness condition, there exists a signal $z\in X$ such that $\{f(x_0),f(u),f(v), f(z)\}$ are not collinear. Consider the timing function  $T_{\{x_0, u, v, z\}}(x_0)=1, T_{\{x_0, u, v, z\}}(u)=1 , T_{\{x_0, u, v, z\}}(v)=2, T_{\{x_0, u, v, z\}}(z)=2$. By applying Lemma \ref{lem_uniq} for the sub-coalition $\{x_0,u,z\}$ and their corresponding timing, shows that the representation holds for $\{u,z\}$ and their timing $T_{\{u, z\}}(u)=1, T_{\{u, z\}}(z)=2$. Then, by considering the coalition $\{u,v,z\}$ and their corresponding timing, Lemma \ref{lem_uniq} shows that the representation holds for $\{u,v\}$ and the timing function $T_{\{u, v\}}(u)=1, T_{\{u, v\}}(v)=2$.

\noindent\textbf{Step 4:} In this step, we show that given any $t\in \mathbb{N}$, the representation holds for any two signals $u,v\in X$ and the timing function $T_{\{u, v\}}(u)=1, T_{\{u, v\}}(v)=t$. The proof is by induction on $t$. By Step 3, the representation holds for $t=2$. Assume that the representation holds for all $t<k$ with $k>3$. We will show that it also holds for $t=k$.

\textbf{Case 1:} If $f(u)\neq f(v)$, then we consider a signal $z\in X$ such that $\{f(u),f(v),f(z)\}$ are not collinear. Let the timing function be  $T_{\{u, v, z\}}(u)=1, T_{\{u, v, z\}}(v)\allowbreak=k+1, T_{\{u, v, z\}}(z)=k$. 

Consider any $w(u,1),w(v,k),w(z,k-1) \in (0,1)$ such that $f(\{u,v,z\},T_{u,v,z})=w(u,1)f(u)+w(v,k)f(v)+w(z,k-1)f(z)$. By Lemma \ref{lem_uniq}, induction hypothesis, and the stationarity axiom, we have $\frac{w(v,k)}{w(u,1)}=\frac{w(v,k)}{w(z,k-1)}\times \frac{w(z,k-1)}{w(u,1)}=(q\frac{w(v)}{w(z)})(q^{k-2}\frac{w(z)}{w(u)})=q^{k-1}\frac{w(v)}{w(u)}$. Thus, the representation holds. 

\textbf{Case 1:} If $f(u)=f(v)$, then we consider two signals $z,z'\in X$ such that \linebreak$\{f(u),f(v),f(z),f(z')\}$ are not collinear (which is possible by the richness condition). Let the timing function be  $T_{\{u, v, z, z'\}}(u)=1, T_{\{u, v, z, z'\}}(v)\allowbreak=k+1, T_{\{u, v, z, z'\}}(z)=k, T_{\{u, v, z, z'\}}(z')=k$. By the uniqueness part of Theorem \ref{main_theorem} and the induction hypothesis, the representation still holds in this case. 

\noindent\textbf{Step 4:} Finally, for any coalition $A\in X^*$ and any timing function $T_A$, the uniqueness of Theorem \ref{main_theorem} and Step 4 establish that the representation should hold with $q,w$.

%Myself: Last part, if all are noncolinear correct, otherwise add a element, then using the technique of my main theorem prove that the representation holds

%%%%%%%%%%%%%%%%%%%%    Proof of Theorem 4

\subsection{Proof of Theorem \ref{thm_consistencyeqwa} and Corollary \ref{EPR=SCA}}

Assume that the aggregation rule $f:X^*\to \mathcal{R}^m$ satisfies the minimal agreement condition and $v\in \mathbb{R}^m$ is the direction on which all agents agree.

Consider two disjoint coalitions $A,B\in X^*$ with the corresponding cardinal utilities $u_A\in U_A$ and $u_B\in U_B$. Assume that $u_{A\cup B}\in U_{A\cup B}$ is a cardinal utility that represents the preference ordering of the union $A\cup B$. If $u_A=u_B$, then the result is trivial. Hence, consider the case $u_A\neq u_B$.

First, by using the Farkas' Lemma, we show that the extended Pareto is equivalent to $u_{A\cup B}\in \text{Cone}^{\circ}(u_A,u_B)$ (which $\text{Cone}^{\circ}(u_A,u_B)$ denotes the interior of the cone generated by $u_A$ and $u_B$). 

If $u_{A\cup B}\in \text{Cone}^{\circ}(u_A,u_B)$, then there exist $\alpha , \beta>0$ such that $u_{A\cup B}=\alpha u_A+\beta u_B$.  Therefore, for any $x,y\in L$ if $u_A\cdot x\geq u_A\cdot y$ and $u_B\cdot x\geq u_B\cdot y$, then $u_{A\cup B}\cdot x\geq u_{A\cup B}\cdot y$. Similarly, if $u_A\cdot x> u_A\cdot y$ and $u_B\cdot x\geq u_B\cdot y$, then $u_{A\cup B}\cdot x> u_{A\cup B}\cdot y$. This proves that the preference ordering of $A\cup B$, satisfies the extended Pareto axiom.\\

For the other side, if the utility of the union $u_{A\cup B}\notin\text{Cone}^{\circ}(u_A,u_B)$, then $\nexists \alpha, \beta>0 $ such that $u_{A\cup B}=\alpha u_A+\beta u_B$. The Farkas' Lemma guarantees that there exists a vector $z\in \mathbb{R}^m$ such that $z\cdot u_A\geq 0, z\cdot u_B\geq 0$ and $z\cdot u_{A\cup B}<0$. 

Consider a vector $y\in \mathbb{R}^m$ that is in the interior of $L$. We select $\lambda>0$ such that $y+\lambda z\in L$. This is possible since we assume that $y$ is in the interior of $L$.
By defining $x=y+\lambda z$, we get $x-y=\lambda z$. Since $\lambda>0$ and $z\cdot u_A\geq 0, z\cdot u_B\geq 0$, we have $u_A\cdot x\geq u_A\cdot y$, and $u_B\cdot x\geq u_B\cdot y$. But since $z\cdot u_{A\cup B}<0$, we have $u_{A\cup B}$ and $x<u_{A\cup B}\cdot y$.
But by the extended Pareto axiom, this cannot be true. Therefore, $u_{A\cup B}\in\text{Con}^{\circ}(u_A,u_B)$. \\

Now consider the intersection of $H=\{x\in\mathbb{R}^m|\ x\cdot v=1\}$ and $\text{Cone}^{\circ}(u_A,u_B)$. Since $u_A\cdot v>0$ and $u_B\cdot v>0$, there should be a unique $
\hat{u}_A\in U_A, \hat{u}_B\in U_B$ both in $H$. It is trivial that $\text{Con}^{\circ}(u_A,u_B)=\text{Con}^{\circ}(\hat{u}_A,\hat{u}_B)$. 

Since both $\hat{u}_A.v>0$ and $\hat{u}_B.v>0$, the intersection of the interior of the cone generated by them and the linear variety $H$ is the segment $[\hat{u}_A,\hat{u}_B]=\{\lambda \hat{u}_A+(1-\lambda)\hat{u}_B| \ \lambda \in (0,1)\}$. Since $u_{A\cup B} \in \text{Con}^{\circ}(\hat{u}_A,\hat{u}_B)$, we should have $v.u_{A\cup B}>0$. Hence, there should be a $\hat{u}_{A\cup B}\in H$ representing $u_{A\cup B}$. Therefore, $\hat{u}_{A\cup B}\in [\hat{u}_A,\hat{u}_B]$. This completes the proof.

\qed

%%%%%%%%%%%%%%%%%%%%    Proof of Theorem 6
\subsection{Proof of Theorem \ref{EPR_rep_main}}

There are a couple of steps in the proof. Note that for any profile $R\in \mathcal{R}_X$, and for any coalition $A\subseteq X$, $R_A$ denotes the restricted sub-profile of the coalition $A$.

\noindent\textbf{Step 1:} Fix a preference $\hat{r}\in \mathcal{R}_{\overline{x}}$. Using the corollary \ref{EPR_rep}, for any profile $R\in \mathcal{R}_X$ such that $R_1=\hat{r}$, we can uniquely define a weight function (which depends on the full profile $R$) $w^{R}(i): X\to \mathbb{R}_{++}$ with $w^{R}(1)=1$, such that for any coalition $A\subseteq X$ we have:
$$
u_H(f(R_A))=\sum\limits_{i\in A}\left(\frac{w^R(i)}{\sum\limits_{j\in A} w^R(j)}\right)u_H(R_i).
$$

First, we show that for any individual $i\in X\setminus\{1\}$ and for any two profiles $R_a,R_b\in \mathcal{R}_X$ with $(R_a)_1=(R_b)_1=\hat{r}$ and $(R_a)_i=(R_b)_i$, we have $w^{R_a}(i)=w^{R_b}(i)$.
There are two separate cases:

\textbf{Case 1:} If $(R_a)_i=(R_b)_i\neq \hat{r}$, then  using $(R_a)_{\{1,i\}}=(R_b)_{\{1,i\}}$ and the result of corollary \ref{EPR_rep}, we should have $w^{R_a}(i)=w^{R_b}(i)$.

%\todo[color=red]{This is true for $N>4$, but for $N=4$ this might be wrong. However,there are two other proxy profiles that we may use them and prove it. the result is in my note.}
\textbf{Case 2:} If $(R_a)_i=(R_b)_i= \hat{r}$, then by considering the definition of the domain $\mathcal{R}_X$, which require the existence of three non-collinear preferences in each profile, there should be a profile $R_c\in \mathcal{R}_X$ and two individual ${j_1},{j_2}\in X\setminus\{1,i\}$ such that $(R_c)_1=(R_c)_i=\hat{r}$,  $\ (R_c)_{j_1}=(R_a)_{j_1}\neq \hat{r}$, and $\ (R_c)_{j_2}=(R_b)_{j_2}\neq \hat{r}$. Using Case 1, we have  $w^{R_c}(j_1)=w^{R_a}(j_1)$ and  $w^{R_c}(j_2)=w^{R_b}(j_2)$. Since $(R_a)_{\{i,j_1\}}=(R_c)_{\{i,j_1\}}$ and $w^{R_a}(j_1)=w^{R_c}(j_1)$, using corollary \ref{EPR_rep}, we should have $w^{R_a}(i)=w^{R_c}(i)$. Similarly, we have $(R_b)_{\{i,j_2\}}=(R_c)_{\{i,j_2\}}$ and $w^{R_b}(j_2)=w^{R_c}(j_2)$. Therefore, we should have $w^{R_b}(i)=w^{R_c}(i)$. Hence, we have $w^{R_a}(i)=w^{R_b}(i)$. \\

By considering  profiles of the form $R\in\mathcal{R}_X$ with $R_1=\hat{r}$, we can define the weight function $w:X\setminus\{1\}\times \mathcal{R}_{\overline{x}}\to \mathbb{R}_{++}$ such that $w(i,R_i)=w^R(i)$ for all $i\in X\setminus\{1\}$. By the result of Step 1, this function is well defined. Moreover, we define $w(1,\hat{r})=1$. 

At this point, for any preference profile $R\in \mathcal{R}_X$ with $R_1=\hat{r}$ and for any coalition $A\subseteq X$, we have:

$$  
u_H(f(R_A))=\sum\limits_{i\in A}\left(\frac{w(i,R_i)}{\sum\limits_{j\in A} w(j,R_j)}\right)u_H(R_i).
$$

\noindent\textbf{Step 2:} We need to define $w(1,r)$ for all $r\in \mathcal{R}_{\overline{x}}$. We have already fixed the value $w(1,\hat{r})=1$. For any $r\in \mathcal{R}_{\overline{x}} \setminus\{\hat{r}\}$,  let $R\in \mathcal{R}_X$ be a profile with $R_1=r$ and $R_2=\hat{r}$. By corollary \ref{EPR_rep}, there should be a unique function $w^R:X\to\mathbb{R}_{++}$ with $w^R(2)=w(2,\hat{r})$. We define $w(1,r)=w^R(1)$. 

 Notice that for any two profile $R_a,R_b\in \mathcal{R}_X$ with $(R_a)_1=(R_b)_1=r$ and $ (R_a)_2=(R_a)_2=\hat{r}$, if we normalize the value of the $w^{R_a}(2)=w^{R_b}(2)=w(2,\hat{r})$, then we should have $w^{R_a}(1)=w^{R_b}(1)$.
Hence, the value $w(1,r)$ is independent of the choice of the profile $R$.

At this point the function $w:X\times\mathcal{R}_{\overline{x}}\to \mathbb{R}_{++}$ is fully defined. We only need to show that it works. 

\noindent\textbf{Step 3:} Select any profile $R\in \mathcal{R}_X$. We need to show that the representation holds with the weight function defined above.

If $R_1=\hat{r}$, by the result of Step 1 the representation holds. Hence, fix any $\overline{r} \in\mathcal{R}_{\overline{x}}$ where $\overline{r}\neq \hat{r}$. In the rest of the proof we show that the representation holds for any $R\in \mathcal{R}_X$ with $R_1=\overline{r}$.

Similar to Step 1, using the corollary \ref{EPR_rep}, for any profile $R\in \mathcal{R}_X$ such that $R_1=\overline{r}$, we can uniquely define a weight function (depending on the full profile $R$) $w'^{R}(i): X\to \mathbb{R}_{++}$ with $w'^{R}(1)=w(1,\overline{r})$, such that for any coalition $A\subseteq X$ we have:
$$
u_H(f(R_A))=\sum\limits_{i\in A}\left(\frac{w'^R(i)}{\sum\limits_{j\in A} w'^R(j)}\right)u_H(R_i).
$$

In the same manner as Step 1, for any two profiles $R_a,R_b\in \mathcal{R}_X$ with $(R_a)_1=(R_b)_1=\overline{r}$ and for every individual $i\in X$, we should have $w'^{R_a}(i)=w'^{R_b}(i)$. Hence, by considering  profiles of the form $R\in\mathcal{R}_X$ with $R_1=\overline{r}$, we can define the weight function $w':X\setminus\{1\}\times \mathcal{R}_{\overline{x}}\to \mathbb{R}_{++}$ such that for all $i\in X\setminus\{1\}:\ w'(i,R_i)=w'^R(i)$. By the result of Step 1, this function is well defined. Moreover, we fix $w'(1,\overline{r})=w(1,r)$.

For every preference profile $R\in \mathcal{R}_X$ with $R_1=\overline{r}$ and for every coalition $A\subseteq X$, we have:

$$  
u_H(f(R_A))=\sum\limits_{i\in A}\left(\frac{w'(i,R_i)}{\sum\limits_{j\in A} w'(j,R_j)}\right)u_H(R_i).
$$

To complete the proof, since we have $w(1,\overline{r})=w'(1,\overline{r})$ , it remains to show that for all $i\in X\setminus\{1\}$ and for all $r\in \mathcal{R}_{\overline{x}}$ we have 
$w(i,r)=w'(i,r)$.

\textbf{Case 1:} Since $\hat{r}\neq\overline{r}$ and $w(1,\hat{r})=w'(1,\hat{r})$, based on Step 2 we should have $w(2,\hat{r})=w'(2,\hat{r})$.

\textbf{Case 2:} Assume that $r\neq \hat{r}$ and $i\in X\setminus\{1,2\}$. Since $N\geq 5$, based on definition of $\mathcal{R}_X$, there exist $R_a,R_b \in \mathcal{R}_X$ such that $(R_a)_1=\hat{r},(R_a)_2=\hat{r},(R_a)_i=r$ and $(R_b)_1=\overline{r},(R_b)_2=\hat{r},(R_b)_i=r$. 

Since $r\neq \hat{r}$,  $(R_a)_{\{2,i\}}=(R_b)_{\{2,i\}}$, and by Case 1 $w(2,\hat{r})=w'(2,\hat{r})$, then we should have $w(i,r)=w'(i,r)$.

\textbf{Case 3:} Assume that $r=\hat{r}$ and $i\in X\setminus\{1,2\}$. Since $N\geq 5$, we can select an individual $j\in X\setminus{\{1,2,i\}}$. Based on the definition of $\mathcal{R}_X$, there exist $R_a,R_b \in \mathcal{R}_X$ such that $(R_a)_1=\hat{r},(R_a)_i=\hat{r},(R_a)_j=\overline{r}$ and $(R_b)_1=\overline{r},(R_b)_i=\hat{r},(R_b)_j=\overline{r}$.

Since $r=\hat{r}\neq \overline{r}$,  $(R_a)_{\{i,j\}}=(R_b)_{\{i,j\}}$, and by Case 2 $w(j,\overline{r})=w'(j,\overline{r})$, then we should have $w(i,r)=w'(i,r)$. 

\textbf{Case 4:} Finally, assume that $i=2$ and $r\neq \hat{r}$. Select an individual $j\in X\setminus{\{1,2\}}$. We consider profiles $R_a,R_b\in \mathcal{R}_X$ such that $(R_a)_1=\hat{r},(R_a)_2=r,(R_a)_j=\hat{r}$ and $(R_b)_1=\overline{r},(R_b)_2=r,(R_a)_j=\hat{r}$. 
By Case 3, we have $w(i,\hat{r})=w'(i,\hat{r})$. Hence, since $r\neq \hat{r}$ and $(R_a)_{\{2,j\}}=(R_b)_{\{2,j\}}$ we should have $w(2,r)=w'(2,r)$.

The last observation completes the proof. 
\qed

\end{document}